\documentclass{article}[11pt]

\usepackage{arxiv}

\usepackage[utf8]{inputenc} 
\usepackage[T1]{fontenc}    
\usepackage{hyperref}       
\usepackage{url}            
\usepackage{booktabs}       
\usepackage{amsfonts}       
\usepackage{nicefrac}       
\usepackage{microtype}      
\usepackage{lipsum}

\usepackage{amsmath}
\usepackage{amsthm}
\usepackage{bm}

\usepackage{psfrag,epsf}
\usepackage{enumerate}
\usepackage{url} 
\usepackage{siunitx}
\usepackage{algorithm}
\usepackage{float}
\usepackage{appendix}
\usepackage[noend]{algpseudocode}
\usepackage{subcaption}
\usepackage{xr}

\usepackage{natbib}
\usepackage{graphicx}

\newtheorem{proposition}{Proposition}
\newtheorem{remark}{Remark}
\newtheorem{lemma}{Lemma}

\usepackage[OT1]{fontenc}

\numberwithin{equation}{section}

\numberwithin{equation}{section}

\DeclareFontFamily{U}{mathx}{\hyphenchar\font45}
\DeclareFontShape{U}{mathx}{m}{n}{
      <5> <6> <7> <8> <9> <10>
      <10.95> <12> <14.4> <17.28> <20.74> <24.88>
      mathx10
      }{}
\DeclareSymbolFont{mathx}{U}{mathx}{m}{n}
\DeclareFontSubstitution{U}{mathx}{m}{n}
\DeclareMathAccent{\widebar}{0}{mathx}{"73}

\def\bx{\bm x}
\def\btheta{\bm \theta}
\def\be{\bm e}
\def\bs{\bm s}
\def\U{\mathcal U}
\def\SU{\mathcal SU}
\def\SN{\mathcal SN}
\def\PN{\mathcal PN}
\def\G{\mathcal{G}}
\def\Q{\mathcal{Q}}
\def\iid{\sim^\text{\it iid}}

\def\glhat{\hat{g}_{\l}}
\def\gfhat{\hat{g}_{\f}}
\def\l{\bm l}
\def\f{\bm f}
\def\h{\bm h}
\def\p{p}

\def\flash{{\it flash}}
\def\ERsq{\widebar{R^2}}
\def\precij{\tau_{ij}}
\def\E{\text{E}}

\def\Ef{\text{E}_{\qf}(f_j)}
\def\Efsq{\text{E}_{\qf}(f^2_j)}
\def\Efk{\text{E}_{\qf}f_{kj}}
\def\Efksq{\text{E}_{\qf}(f^2_{kj})}

\def\El{\text{E}_{\ql}(l_i)}
\def\Elsq{\text{E}_{\ql}(l^2_i)}

\def\gfk{g_{\f_k}}
\def\glk{g_{\l_k}}
\def\ghk{g_{\h_k}}

\def\Gl{\G_{\l}}
\def\Gf{\G_{\f}}

\def\ll{\widebar{\l}}
\def\ff{\widebar{\f}}
\def\llsq{\widebar{\l^2}}
\def\ffsq{\widebar{\f^2}}

\def\ebnm{\text{\it EBNM}}
\def\prec{\bm \tau}
\def\T{\mathcal T}
\def\iter{t}

\def\ql{q_{\l}}
\def\qf{q_{\f}}
\def\gl{g_{\l}}
\def\gf{g_{\f}}

\begin{document}

\title{Empirical Bayes Matrix Factorization}


\author{Wei Wang \\
      Department of Statistics \\
       University of Chicago \\
       Chicago, IL, USA \\
       {\texttt weiwang@statistics.uchicago.edu} \\
       \And
        Matthew Stephens \\
        Department of Statistics and Department of Human Genetics \\
       University of Chicago \\
       Chicago, IL, USA \\
       {\texttt mstephens@uchicago.edu} }

\maketitle

\begin{abstract}
Matrix factorization methods, which include Factor analysis (FA) and Principal Components Analysis (PCA), are widely used for inferring and summarizing structure in multivariate data.  Many such methods use 
a penalty or prior distribution to achieve sparse representations (``Sparse FA/PCA"), and a key question
is how much sparsity to induce. Here we introduce a general Empirical Bayes approach to matrix factorization (EBMF), whose key feature is that it estimates the appropriate amount of sparsity by estimating prior distributions from the observed data. The approach is very flexible: it allows for a wide range of different prior families and allows that each component of the matrix factorization may exhibit a different amount of sparsity. The key to this flexibility is the use of a variational approximation, which we show effectively reduces fitting the EBMF model to solving a simpler problem, the so-called ``normal means" problem.  We demonstrate the benefits of 
EBMF with sparse priors through both numerical comparisons with competing methods and through analysis of data from the 
GTEx (Genotype Tissue Expression) project on genetic associations across 44 human tissues. In numerical comparisons EBMF often provides more accurate inferences than other methods. In the GTEx data, EBMF identifies interpretable structure that agrees with known relationships among human tissues. Software implementing our approach is available at \url{https://github.com/stephenslab/flashr}.
\end{abstract}


\section{Introduction}

Matrix factorization methods are widely used for inferring and summarizing structure in multivariate data.
In brief, these methods represent an observed $n \times p$ data
matrix $Y$ as:
\begin{equation} \label{eqn:fa}
Y = L^T F + E
\end{equation}
where $L$ is an $K \times n$ matrix, $F$ is a $K \times p$ matrix, and $E$ is an $n \times p$ matrix of residuals (whose entries we assume to be normally distributed, although the methods we develop can be generalized to other settings; see Section \ref{sec:nongauss}). Here we adopt the notation and terminology of factor analysis, and refer to $L$ as the ``loadings'' and $F$ as the ``factors''.

The model \eqref{eqn:fa} has many potential applications.
One range of applications arise from ``matrix completion" problems \citep[e.g.][]{fithian2018flexible}:
methods that estimate $L$ and $F$ in \eqref{eqn:fa} from partially observed $Y$ provide a natural and effective way to fill in the missing entries. 
Another wide range of applications come from the desire
to {\it summarize and understand} the structure in 
a matrix $Y$: in \eqref{eqn:fa}
each row of $Y$ is approximated by a linear combination of 
underlying ``factors" (rows of $F$), which---ideally---have some  intuitive or scientific interpretation. For example, suppose $Y_{ij}$  represents the rating of a user $i$ for a movie $j$. Each factor might represent a genre of movie (``comedy", ``drama", ``romance", ``horror'' etc), 
and the ratings for a user $i$ could be written as a linear combination
of these factors, with the weights (loadings) representing how much individual
$i$ likes that genre. Or, suppose $Y_{ij}$ represents the expression of
gene $j$ in sample $i$. Each factor might represent a module of co-regulated genes,
and the data for sample $i$ could be written as a linear combination of these factors,
with the loadings representing how active each module is in each sample. Many other examples 
could be given across many fields, including psychology \citep{ford1986application}, econometrics \citep{bai2008large}, natural language processing \citep{bouchard-etal-2015-matrix}, population genetics \citep{Engelhardt2010}, and functional genomics \citep{stein2018enter}.

The simplest approaches to estimating $L$ and/or $F$ in \eqref{eqn:fa} are based on maximum likelihood or least squares.  For example, Principal Components Analysis
(PCA)---or, more precisely, truncated Singular Value Decomposition (SVD)---can be interpreted as fitting \eqref{eqn:fa} by least squares, assuming that columns of $L$ are orthogonal and columns of $F$ are orthonormal \citep{eckart.young.1936}. And classical factor analysis (FA) corresponds to maximum likelihood estimation of $L$, assuming that the elements of $F$ are independent standard normal and allowing different residual variances for each column of $Y$ \citep{rubin1982algorithms}. While
these simple methods remain widely used, in the last two decades
researchers have focused considerable attention on obtaining more accurate and/or more interpretable estimates, either by imposing additional constraints  \citep[e.g.~non-negativity;][]{lee:1999} or by regularization using a penalty term \citep[e.g.][]{jolliffe2003modified,Witten2009,mazumder2010spectral,hastie2015matrix,fithian2018flexible}, or a prior distribution \citep[e.g.][]{bishop1999variational,attias1999independent,ghahramani2000variational,West2003}. In particular, many authors have noted the benefits of sparsity assumptions on
$L$ and/or $F$---particularly in applications where interpretability of the estimates is desired---and there now exists a wide range of methods that attempt to induce sparsity in these models \citep[e.g.][]{sabatti2005bayesian,H.Zou2006,Pournara2007,Carvalho2008,Witten2009,Engelhardt2010,Knowles2011,bhattacharya2011sparse,mayrink2013sparse,yang2014sparse,10.1371/journal.pcbi.1004791,hore2016tensor,rovckova2016fast,srivastava2017expandable,kaufmann2017identifying,fruhwirth2018sparse,zhao2018fast}.  
 Many of these methods induce sparsity in the loadings only, although some induce sparsity in both loadings and factors. 
 
In any statistical problem involving sparsity, 
a key question is how strong the sparsity should be. 
In penalty-based methods this is controlled by the strength 
and form of the penalty, whereas in Bayesian methods
it is controlled by the prior distributions.
In this paper we take an Empirical Bayes approach to this problem, exploiting variational approximation methods \citep{Blei2016} to obtain simple algorithms that jointly estimate
the prior distributions for both loadings and factors,
as well as the loadings and factors themselves.

Both EB and variational methods have been previously
used for this problem \citep{bishop1999variational,lim2007variational,raiko2007principal, stegle2010bayesian}. However, most of this previous work
has used simple normal prior distributions that do not induce sparsity. Variational methods that use sparsity-inducing priors include \cite{girolami2001variational}, which uses a Laplace prior on the factors (no prior on the loadings which are treated as free parameters), \cite{hochreiter2010FABIA} which extends this to Laplace priors on both factors and loadings, with fixed values of the Laplace prior parameters; \cite{titsias2011spike},
which uses a sparse ``spike-and-slab" (point normal) prior on the loadings (with the same prior on all $K$ loadings) and a normal prior on the factors; and \cite{hore2016tensor} which uses a spike-and-slab prior
on one mode in a tensor decomposition. (While this
work was in review further examples appeared, including
\citealt{argelaguet2018multi}, which uses normal priors on the loadings and point-normal on the factors.)

Our primary contribution here is to develop and implement a more general EB approach to matrix factorization (EBMF).
This general approach allows for a wide range of potential sparsity-inducing prior distributions on both the loadings and the factors within a single algorithmic framework. 
We accomplish this by
showing that, when using variational methods, fitting EBMF with {\it any} prior family
can be reduced to repeatedly solving a much simpler problem---the ``empirical Bayes normal means" (EBNM) problem---with the same prior family. This feature makes it easy to implement methods for any desired prior family---one simply has to implement a method
to solve the corresponding normal means problem, and then plug this into our algorithm.  This approach can work for both
parametric families (e.g.~normal, point-normal, laplace, point-laplace) and non-parametric families, including the ``adaptive shrinkage" priors (unimodal and scale mixtures of normals) from
 \cite{stephens:2017}. It is also possible to
 accommodate non-negative constraints on either $L$ and/or $F$
 by using non-negative prior families. 
  Even simple versions of our approach---e.g.~ using
  point-normal priors on both factors and loadings--- provide more generality than most existing EBMF approaches and software. 
 
 A second contribution of our work is to highlight similarities and differences between EBMF and penalty-based methods for regularizing $L$ and/or $F$.
  Indeed, our algorithm for fitting EBMF
 has the same structure as commonly-used algorithms 
 for penalty-based methods, 
 with the prior distribution playing a role analogous
to the penalty (see Remark \ref{algorithm_remark} later). While the general correspondence between
estimates from penalized methods and Bayesian posterior modes
(MAP estimates) is well known, the connection here is different, because the EBMF approach is estimating a posterior mean, not a mode (indeed, with sparse priors the MAP estimates of $L$ and $F$ are not useful because they are trivially 0). A key difference between the EBMF approach
and penalty-based methods is that the EBMF prior is estimated
by solving an optimization problem, whereas in penalty-based methods the strength of the penalty is usually chosen by cross-validation. This difference makes it much easier for
EBMF to allow for different levels of sparsity in every factor and every loading: in EBMF one simply uses a different prior for every factor and loading, whereas tuning a separate parameter
for every factor and loading by CV becomes very cumbersome.

The final contribution is that we provide an R software package, \flash{} (Factors and Loadings by Adaptive SHrinkage), implementing our flexible EBMF framework.
We demonstrate the utility of these methods through both numerical comparisons with competing methods and through a scientific application: analysis of data from the  GTEx (Genotype Tissue Expression) project on genetic associations across 44 human tissues. In numerical comparisons \flash{} often provides more accurate inferences than other methods, while remaining computationally tractable for moderate-sized matrices (millions of entries). In the GTEx data, \flash{} highlight both effects that are shared across many tissues (``dense" factors) and effects that are specific to a small number of tissues (``sparse" factors). These sparse factors often highlight similarities between tissues that are known to be biologically related, providing external support for the reliability of the results.





\section{A General Empirical Bayes Matrix Factorization Model} \label{sec:EBMF}

We define the $K$-factor Empirical Bayes Matrix Factorization (EBMF) model as follows: 
\begin{align} \label{eqn:EBMF}
Y  &= \sum_{k =1}^K \l_k \f^T_k + E \\ \label{eqn:EBMF-l}
l_{k1},\dots,l_{kn} &\iid \glk, \quad \glk \in \Gl \\ \label{eqn:EBMF-f}
f_{k1},\dots,f_{kp} &\iid \gfk, \quad \gfk \in \Gf \\ 
\label{eqn:EBMF-e}
E_{ij} &\sim N(0,1/\tau_{ij}) \text{ with } \prec :=(\tau_{ij}) \in \T. 
\end{align}
Here $Y$ is the $n \times p$ observed data matrix, $\l_k$ is an $n$-vector (the $k$th set of ``loadings"), $\f_k$ is a $p$-vector (the $k$th ``factor"), $\Gl$ and $\Gf$  are pre-specified (possibly non-parametric) families of distributions, $\glk$ and $\gfk$ are unknown ``prior" distributions that are to be estimated, 
$E$ is an $n \times p$ matrix of independent error terms,
and $\prec$ is an unknown $n \times p$ matrix of precisions ($\tau_{ij}$)
which is assumed to lie in some space $\T$. (This allows structure to be imposed on $\prec$, such as constant precision, $\tau_{ij}=\tau$, or column-specific precisions, $\tau_{ij} = \tau_j$, for example.) 
Our methods allow that some elements of $Y$ may 
be ``missing", and can estimate the missing values (Section \ref{sec:missing}). 

The term ``Empirical Bayes" in EBMF means we
fit \eqref{eqn:EBMF}-\eqref{eqn:EBMF-e}
by obtaining {\it point estimates} for the priors $\glk,\gfk$, $(k=1,\dots,K)$ and approximate the posterior distributions
for the parameters $\l_k,\f_k$ given those point estimates. This contrasts with a ``fully Bayes" approach that, instead of obtaining point estimates for $\glk,\gfk$, would 
integrate over uncertainty in the estimates. 
This would involve specifying prior distributions for $\glk,\gfk$ as well as (perhaps substantial) additional computation. The EB approach has the advantage of simplicity --
both conceptually and computationally---while enjoying many
of the benefits of a fully Bayes approach.
In particular it allows for sharing of information across
elements of each loading/factor. For example, if the data suggest that a particular factor, $\f_k$, is sparse, then this will be reflected in a sparse estimate of $\gfk$, and subsequently strong shrinkage of the smaller elements of $f_{k1},\dots,f_{kp}$ towards 0. 
Conversely, when the data suggest a non-sparse
factor then the prior will be dense and the shrinkage less strong. By allowing different prior distributions for each factor and each loading, 
the model has the flexibility to adapt to any combination of sparse and dense loadings and factors. 
However, to fully capitalize on this flexibility one needs suitably flexible prior families
$\Gl$ and $\Gf$ capable of capturing both sparse and dense factors. A key feature of our work is it allows for very flexible prior families, including non-parametric families.

Some specific choices of the distributional families $\Gl$ and $\Gf$ correspond to models used in previous work.
In particular, many previous papers have
studied the case with normal priors, where $\Gl$ and $\Gf$ are both the family of zero-mean normal distributions \citep[e.g.][]{bishop1999variational, lim2007variational,raiko2007principal,nakajima2011theoretical}. This family is particularly simple, having a single hyper-parameter, the prior variance, to estimate for each factor. However, it does not induce sparsity on either $L$ or $F$; indeed, when the matrix $Y$ is fully observed, the estimates of $L$ and $F$ under a normal prior (when using a fully factored variational approximation) are simply scalings of the singular vectors from an SVD of $Y$ \citep{nakajima2011theoretical,nakajima2013global}.  Our work here extends these previous approaches to a much wider range of prior families that {\it do} induce sparsity on $L$ and/or $F$.

We note that the EBMF model \eqref{eqn:EBMF}-\eqref{eqn:EBMF-e} differs
in an important way
from the sparse factor analysis (SFA) methods
in \cite{Engelhardt2010}, which
use a type of Automatic Relevance Determination prior \citep[e.g.][]{tipping2001sparse,wipf2008new} to induce sparsity on the loadings matrix.
In particular, SFA estimates
a separate hyperparameter for every element of the loadings matrix, with no sharing of information
across elements of the same loading. In contrast,
EBMF estimates a single {\it shared} prior distribution for elements of each loading,
which, as noted above, allows for
sharing of information across elements
of each loading/factor. 

\section{Fitting the EBMF Model}

To simplify exposition we begin with the case $K=1$ (``rank 1"); see Section \ref{sec:KEBMF} for the extension to general $K$. To simplify notation we
assume the families $\Gl,\Gf$ are the same, so we can write $\Gl=\Gf =\G$. To further lighten notation in the case $K=1$ we use $\gl, \gf, \l, \f$ instead of ${\gl}_1,{\gf}_1,\l_1,\f_1$. 

Fitting the EBMF model involves estimating all of $\gl, \gf, \l,\f,\prec$. A standard EB approach would be to do this in two steps: 
\begin{itemize}
\item Estimate $\gl,\gf$ and $\prec$, by maximizing the likelihood:
\begin{equation} \label{eqn:lik2}
L(\gl,\gf,\prec) := \int \int p(Y | \l, \f, \prec) \, \gl(dl_1)\dots \gl(dl_n) \, \gf(df_1)\dots \gf(df_p)
\end{equation}
over $\gl,\gf \in \G, \prec \in \T$. (This optimum will typically not be unique because of identifiability issues; see Section \ref{sec:identifiability}.)
\item Estimate $\l$ and $\f$ using their posterior distribution: $p(\l,\f | Y, \glhat,\gfhat,\hat{\prec})$.
\end{itemize}
However, both these two steps are difficult, even for very simple choices of $\G$. Instead, following previous work (see Introduction for citations) we use variational approximations 
to approximate this approach. Although variational approximations are known to typically under-estimate uncertainty in posterior distributions, our focus here is on obtaining useful point estimates for $\l,\f$; results shown later demonstrate that the variational approximation can perform well in this task.

\subsection{The Variational Approximation}
 
The variational approach---see \citet{Blei2016} for review---begins by writing the log of the likelihood \eqref{eqn:lik2} as:
\begin{align}
\label{log-likelihood}
l(\gl, \gf, \prec)  &:= \log L(\gl,\gf,\prec) \\
& = F(q,\gl,\gf,\prec) + D_{KL}(q||p)
\end{align}
where
\begin{equation}
\label{elbow}
F(q,\gl,\gf,\prec) = \int q(\l,\f)\log \frac{p(Y,\l,\f | \gl,\gf,\prec)}{q(\l,\f)} \, d\l \, d\f,
\end{equation}
and 
\begin{equation}
\label{KLdistance}
D_{KL}(q||p) = -\int  q(\l,\f) \log \frac{p(\l,\f |Y,\gl,\gf,\prec) }{q(\l,\f)} \, d\l \, d\f
\end{equation}
is the Kullback--Leibler divergence from $q$ to $p$. This identity holds for any distribution $q(\l,\f)$.
Because $D_{KL}$ is non-negative, it follows that $F(q,\gl,\gf,\prec)$ is a lower bound for the log likelihood:
\begin{equation}
l(\gl,\gf,\prec) \geq F(q,\gl,\gf,\prec)
\end{equation} 
with equality when $q(\l,\f) = p(\l,\f|Y,\gl,\gf,\prec)$.

In other words,
\begin{equation} \label{eqn:log-lik-max}
l(\gl, \gf, \prec) = \max_{q} F(q, \gl, \gf, \prec),
\end{equation}
where the maximization is over all possible distributions $q(\l,\f)$.
Maximizing $l(\gl, \gf, \prec)$ can thus be viewed as maximizing $F$ over $q, \gl, \gf, \prec$.
However, as noted above, this maximization is difficult.
The variational approach simplifies the problem by maximizing $F$ but restricting the family of distributions for $q$.
Specifically, the most common variational approach---and the one we consider here---restricts $q$ to the family $\Q$ of distributions that ``fully-factorize": 
\begin{equation}
\Q = \left\{q : q(\l,\f) = \prod_{i=1}^n q_{l,i}(l_i) \prod_{j=1}^p q_{f,j}(f_j)\right\}.
\end{equation}
The variational approach seeks to optimize $F$ over $q, \gl, \gf, \prec$
with the constraint $q \in Q$. 
For $q \in Q$ we can write $q(\l,\f) = \ql(\l) \qf(\f)$ where $\ql(\l) = \prod_{i=1}^n q_{l,i}(l_i)$ and $\qf(\f)= \prod_{j=1}^p q_{f,j}(f_j)$,
and we can consider the problem as maximizing $F(\ql, \qf, \gl, \gf, \prec)$.

\subsection{Alternating Optimization}

We optimize $F(\ql,\qf,\gl,\gf,\prec)$ by alternating 
between optimizing over variables 
related to $\l$ [$(\ql,\gl)$], over variables related to $\f$ [$(\qf,\gf)$], and over $\prec$. Each of these steps is guaranteed
to increase (or, more precisely, not decrease) $F$, and convergence can be assessed by (for example) stopping when these optimization steps yield a very small increase in $F$. Note that $F$ may be multi-modal, and there is no guarantee that the algorithm will converge to a global optimum. The approach is summarized in Algorithm \ref{alg:r1}.

\begin{algorithm}[H]
\caption{Alternating Optimization for EBMF (rank 1)} \label{alg:r1}
\begin{algorithmic}[1]
\Require Initial values $\ql^{(0)},\qf^{(0)},\gl^{(0)}, \gf^{(0)}$ 
\State $\iter \gets 0$
\Repeat
\State $\iter \gets \iter + 1$
\State $\prec^{(\iter)} \gets \arg \max_{\prec} F(\ql^{(\iter-1)}, \qf^{(\iter-1)},  \gl^{(\iter-1)},\gf^{(\iter-1)},\prec)  $
\State $\ql^{(\iter)}, \gl^{(\iter)} \gets \arg \max_{\ql, \gl}F(\ql, \qf^{(\iter - 1)},\gl,\gf^{ (\iter - 1)},\prec^{(\iter)})  $.
\State $\qf^{(\iter)}, \gf^{(\iter)} \gets \arg \max_{\qf, \gf}F(\ql^{(\iter)}, \qf,  \gl^{(\iter)},\gf,\prec^{(\iter)})  $.
\Until{converged} \\
\Return $\ql^{(\iter)}, \qf^{(\iter )},\gl^{(\iter)},\gf^{ (\iter )},\prec^{(\iter)}$
\end{algorithmic}
\end{algorithm}

The key steps in Algorithm \ref{alg:r1} are the maximizations in Steps 4-6. 

Step 4, the update of $\prec$, involves computing the expected squared residuals:
\begin{align} \label{eqn:R2r1}
\ERsq_{ij} &:= \E_{\ql,\qf}[(Y_{ij}-l_i f_j)^2] \\
& = [Y_{ij} - \El \Ef]^2 - \El^2 \Ef^2 + \Elsq \Efsq.
\end{align}
This is straightforward provided the first and second moments of $\ql$ and $\qf$ are available
(see Appendix \ref{sec:est_var} for details).

Steps 5 and 6 are essentially identical except for switching the role of $\l$ and $\f$. One of our key results is that each of these steps can be achieved by solving
a simpler problem---the Empirical Bayes normal means (EBNM) problem. The next subsection (\ref{sec:ebnm}) describes
the EBNM problem, and the following subsection
(\ref{sec:ebmf_and_ebnm}) details how this can be used to solve Steps 5 and 6.

\subsection{The EBNM Problem} \label{sec:ebnm}

Suppose we have observations $\bx = (x_1,\dots,x_n)$ of underlying quantities $\btheta=(\theta_1,\dots,\theta_n)$,
with independent Gaussian errors with known standard deviations $\bs = (s_1,\dots,s_n)$. 
Suppose further that the elements of $\btheta$  are assumed i.i.d. from some distribution, $g \in \G$.
That is,
 \begin{align} \label{eqn:normalmeans}
\bx | \btheta & \sim N_n(\btheta, \text{diag}(s_1^2,\dots,s_n^2)) \\ \label{eqn:theta}
\theta_1,\dots,\theta_n & \iid g, \quad g \in \G,
\end{align}
where $N_n(\bm{\mu},\Sigma)$ denotes the $n$-dimensional normal distribution with mean $\bm{\mu}$ and covariance matrix $\Sigma$. 

By solving the EBNM problem we mean fitting the model \eqref{eqn:normalmeans}-\eqref{eqn:theta}
by the following two-step procedure:
\begin{enumerate}
\item Estimate $g$ by maximum (marginal) likelihood:
\begin{equation} \label{eqn:ghat}
\hat{g} = \arg \max_{g \in \G} \prod_j \int p(x_j | \theta_j, s_j) g(d\theta_j). 
\end{equation} 
\item Compute the posterior distribution for $\btheta$ given $\hat{g}$, 
\begin{equation} \label{eqn:post}
\p(\btheta | \bx, \bs, \hat{g}) \propto \prod_j \hat{g}(\theta_j) p(x_j | \theta_j, s_j).
\end{equation}
Later in this paper we will have need for the posterior first and second moments, so we define them here for convenience:
\begin{align} \label{eqn:postmean}
\bar{\theta}_j &:= E(\theta_j | \bx, \bs, \hat{g})  \\ \label{eqn:postmean2}
\widebar{\theta^2}_j &:= E(\theta^2_j | \bx, \bs, \hat{g}).
\end{align}
\end{enumerate}
Formally, this procedure defines a mapping
(which depends on the family $\G$) from the known quantities $(\bx,\bs)$, to $(\hat{g},\p)$, where $\hat{g},\p$ are given in \eqref{eqn:ghat} and \eqref{eqn:post}. 
We use $\ebnm$ to denote this mapping:
\begin{equation} \label{eqn:ebnm}
\ebnm(\bx,\bs) = (\hat{g},\p).
\end{equation}

\begin{remark} \label{ebnm_remark}
Solving the EBNM problem is central to
all our algorithms, so it is worth some study. A key point is that the EBNM problem provides an attractive and flexible way to induce shrinkage and/or sparsity in estimates of $\theta$. For example, if $\theta$ is truly sparse, with many elements at or near 0, then the estimate $\hat{g}$ will typically have considerable mass near 0, and the posterior means \eqref{eqn:postmean} will be ``shrunk'' strongly toward 0 compared with the original observations. 
In this sense solving the EBNM problem
can be thought of as 
a model-based analogue of thresholding-based methods, with the advantage that
by estimating $g$ from the data the EBNM approach automatically adapts to provide an appropriate level of shrinkage. These ideas have been used in wavelet denoising \citep{Clyde2000,johnstone2004needles,Johnstone2005Empirical,xing2016smoothing},
and false discovery rate estimation \citep{thomas:1985,stephens:2017} for example. 
Here we apply them to matrix factorization problems. 
\end{remark}

\subsection{Connecting the EBMF and EBNM Problems} \label{sec:ebmf_and_ebnm}

The EBNM problem is well studied, and can be solved reasonably easily for many choices of $\G$
\citep[e.g.][]{johnstone2005ebayesthresh,koenker2014JASA,stephens:2017}. In Section \ref{sec:implementation} we give specific examples; for now our main
point is that {\it if} one
can solve the EBNM problem for a particular
choice of $\G$ then it can be used to implement
Steps 5 and 6 in Algorithm \ref{alg:r1} for the
corresponding EBMF problem. The following Proposition formalizes this for
Step 5 of Algorithm \ref{alg:r1};
a similar proposition holds for Step 6 (see also Appendix \ref{app:ebfa_k}).

\begin{proposition} \label{prop:EBNM}
Step 5 in Algorithm~\ref{alg:r1} is solved by solving an EBNM problem. Specifically 
\begin{equation}
\arg \max_{\ql, \gl}F(\ql, \qf,\gl,\gf,\prec)   =  \ebnm (\hat{\l}(Y,\ff,\ffsq,\prec),\bs_l(\ffsq,\prec))
\end{equation}
where
the functions $\hat{\l}: R^{n \times p} \times R^p \times R^p \times R^{n \times p} \rightarrow R^n$ and $\bs_l:R^p \times R^{n \times p} \rightarrow R^n$ are given by
\begin{align} \label{eqn:lhatfn}
\hat{\l}(Y,\bm{v},\bm{w},\prec)_i &:=  \frac{\sum_j \precij Y_{ij} v_j}{\sum_j \precij w_j}, \\ \label{eqn:slfn}
\bm{s}_{\l}(\bm{w},\prec)_i &:= \left(\sum_j \precij w_j \right)^{-0.5},
\end{align}
and  $\ff, \ffsq \in R^p$ denote the vectors whose elements are the first and second moments of $\f$ under $\qf$:
\begin{align}
\ff & := (\text{E}_{\qf}(f_j)) \\
\ffsq &:= (\text{E}_{\qf}(f_j^2)).
\end{align}
\end{proposition}
\begin{proof}
See Appendix \ref{app:ebfa_k}.
\end{proof}


For intuition into where the EBNM in Proposition
\ref{prop:EBNM} comes from, consider
estimating $\l,\gl$ in \eqref{eqn:EBMF}
{\it with $\f$ and $\prec$ known}. 
The model then becomes $n$ independent regressions of the rows of $Y$ on $\f$, and the maximum likelihood estimate for $\l$ has elements:
\begin{equation} \label{eqn:lhat}
\hat{l}_i = \frac{\sum_j \precij Y_{ij}f_j}{\sum_j \precij f^2_j},
\end{equation}
with standard errors 
\begin{equation} \label{eqn:sl}
s_i = \left(\sum_j \precij f_j^2 \right)^{-0.5}.
\end{equation}
Further, it is easy to show that
\begin{equation} \label{eqn:lhat-dist}
\hat{l}_i \sim N(l_i,s^2_i).
\end{equation}

Combining \eqref{eqn:lhat-dist} with the prior
\begin{equation} \label{eqn:lprior}
l_1,\dots,l_n \iid \gl, \quad \gl \in \G 
\end{equation}
yields an EBNM problem. 

The EBNM in Proposition \ref{prop:EBNM} is 
the same as the EBNM \eqref{eqn:lhat-dist}-\eqref{eqn:lprior} ,
but with the terms $f_j$ and $f_j^2$ replaced with their 
expectations under $\qf$. Thus, 
the update for $(\ql,\gl)$ in Algorithm 1, with $(\qf,\gf,\prec)$ fixed, is closely connected to solving the EBMF problem for ``known $\f,\prec$''.

\subsection{Streamlined Implementation Using First and Second Moments}

Although Algorithm \ref{alg:r1}, as written, optimizes over $(\ql, \qf, \gl, \gf)$, in practice each step requires only
the first and second moments of the distributions $\ql$ and $\qf$. For example, the EBNM problem in Proposition 1 involves $\ff$ and $\ffsq$ and not $\gf$. Consequently, we can simplify implementation by keeping track of only those moments.
In particular, when solving the normal means
problem, $\ebnm(\bx,\bs)$ in 
\eqref{eqn:ebnm}, we need only return the posterior first
and second moments \eqref{eqn:postmean} and \eqref{eqn:postmean2}.
This results in a streamlined and intuitive implementation, summarized in Algorithm \ref{alg:r1b}.

\begin{algorithm}[H]
\caption{Streamlined Alternating Optimization for EBMF (rank 1)} \label{alg:r1b}
\begin{algorithmic}[1]
\Require A data matrix $Y$ ($n \times p$) 
\Require A function, ${\tt ebnm}(\bx,\bs) \rightarrow (\widebar{\btheta},\widebar{\btheta^2})$, that solves the EBNM problem \eqref{eqn:normalmeans}-\eqref{eqn:theta} and returns the first and second posterior moments \eqref{eqn:postmean}-\eqref{eqn:postmean2}.
\Require A function, ${\tt init}(Y) \rightarrow (\hat{\l},\hat{\f})$ that produces initial estimates for $\l$ (an $n$ vector) and $\f$ (a $p$ vector) given data $Y$. (For example, rank 1 singular value decomposition.)
\State Initialize first moments $(\ll,\ff)$, using $(\ll,\ff) \gets {\tt init}(Y)$
\State Initialize second moments $(\llsq,\ffsq)$, by squaring first moments: $\llsq \gets (\ll_i^2)$ and $\ffsq \gets (\ff_j^2)$. 
\Repeat 
\State Compute the matrix of expected squared residuals $\ERsq_{ij}$ from \eqref{eqn:R2r1}.
\State $\tau_j \gets n/\sum_i \widebar{R^2}_{ij}$. [This update assumes column-specific variances; it can be modified to make other assumptions.]
\State Compute $\hat{\l}(Y,\ff,\ffsq,\prec)$ and standard errors $\bs_l(\llsq,\prec)$, using \eqref{eqn:lhatfn} and \eqref{eqn:slfn}.
\State $(\ll,\llsq) \gets {\tt ebnm}(\hat{\l},\bs_l)$.
\State Compute $\hat{\f}(Y,\ll,\llsq,\prec)$ and standard errors $\bs_f(\llsq,\prec)$ (similarly as for $\hat{l}$ and $\bs_l$; see \eqref{eqn:fhatfn} and \eqref{eqn:sffn}). 
\State $(\ff,\ffsq) \gets {\tt ebnm}(\hat{\f},\bs_f)$.
\Until converged\\
\Return $\ll,\llsq,\ff,\ffsq,\prec$
\end{algorithmic}
\end{algorithm}

\begin{remark} \label{algorithm_remark}
Algorithm \ref{alg:r1b} has a very intuitive form:
it has the flavor of an alternating least squares algorithm, which alternates between estimating $\l$ given $\f$ (Step 6) and $\f$ given $\l$ (Step 8), but with the addition of the {\tt ebnm} step (Steps 7 and 9), which can be thought of as
regularizing or shrinking the estimates: see Remark \ref{ebnm_remark}. This viewpoint highlights connections with related algorithms. For example,  the (rank 1 version of the) SSVD algorithm from \cite{yang2014sparse} has a similar form, but 
uses a thresholding function in place of the {\tt ebnm} function to induce shrinkage and/or sparsity. 
\end{remark}

\subsection{The \texorpdfstring{$K$}{K}-factor EBMF Model} \label{sec:KEBMF}

It is straightforward to extend the variational approach 
to fit the general $K$ factor model \eqref{eqn:EBMF}-\eqref{eqn:EBMF-e}. In brief, we introduce variational
distributions $(q_{\l_k},q_{\f_k})$ for $k=1,\dots,K$,
and then optimize the objective function $F(q_{\l_1},g_{\l_1},q_{\f_1},g_{\f_1}; \dots;q_{\l_K},g_{\l_K},q_{\f_K},g_{\f_K};\prec)$. 
Similar to the rank-1 model, this optimization can be done
by iteratively updating parameters relating to a
single loading or factor, keeping other parameters fixed.
And again we simplify implementation by keeping track of only the first and second moments of the distributions $q_{\l_k}$ and $q_{\f_k}$, which we denote $\ll_k,\llsq_k,\ff_k,\ffsq_k$. 
The updates to $\ll_k,\llsq_k$ (and $\ff_k,\ffsq_k$)
are essentially identical to those for fitting the rank 1 model above, but with $Y_{ij}$ replaced with the residuals obtained by removing the estimated effects of the other $k-1$ factors:
\begin{equation} \label{eqn:resid}
R^k_{ij} := Y_{ij} - \sum_{k'\neq k} \ll_{k'i} \ff_{k'j}.
\end{equation}

Based on this approach we have implemented two algorithms for fitting the $K$-factor model. First, a simple ``greedy'' algorithm, which starts by fitting the rank 1 model, and then adds factors $k=2,\dots,K$, one at a time, optimizing over the new factor parameters before moving on to the next factor.  Second, a ``backfitting'' algorithm \citep{Breiman1985}, which iteratively refines
the estimates for each factor given the estimates for the other factors. Both algorithms are detailed in Appendix \ref{app:ebfa_k}.

\subsection{Selecting \texorpdfstring{$K$}{K}} \label{r1vsr0}

An interesting feature of EB approaches to matrix factorization, noted by \cite{bishop1999variational}, is that they automatically select the number of factors $K$. This is because
the maximum likelihood solution to $\glk,\gfk$ is sometimes 
a point mass on 0 (provided $\G$ includes this distribution). Furthermore, the same is true of the solution to the variational approximation \citep[see also][]{bishop1999variational,stegle2012using}. This means that if $K$ is set sufficiently large then some loading/factor combinations will be optimized to be exactly 0. (Or, in the greedy approach, which adds one factor at a time, the algorithm will eventually add a factor that is exactly 0, at which point it terminates.)

Here we note that the variational approximation may be expected to result in conservative estimation (i.e. underestimation) of $K$ compared with the (intractable) use of maximum likelihood to estimate $\gl,\gf$. We base our argument on the simplest case: comparing $K=1$ vs $K=0$.
Let $\delta_0$ denote the degenerate distribution with all its mass at 0.
Note that the rank-1 factor model \eqref{eqn:EBMF}, with $\gl = \delta_0$ (or $\gf=\delta_0$)
is essentially a ``rank-0" model. 
Now note that the variational lower bound, $F$, is exactly 
equal to the log-likelihood when $\gl=\delta_0$ (or $\gf=\delta_0$). 
This is because if the prior is a point mass at 0 then the posterior is also a point mass,
which trivially factorizes as a product of point masses, and so the variational family $\Q$ includes the true posterior in this case. Since $F$ is a lower bound to the log-likelihood we have the following simple lemma:
\begin{lemma}
If $F(\hat{q},\glhat, \gfhat,\hat\tau)
> F(\delta_0,\delta_0,\delta_0,\hat\tau_0)$ then $l(\glhat,\gfhat,\hat\tau) > l(\delta_0,\delta_0,\hat\tau_0)$.  
\end{lemma}
\begin{proof}
\begin{equation}
l(\glhat,\gfhat,\hat\tau) \geq F(\hat{q},\glhat, \gfhat,\hat\tau)
> F(\delta_0,\delta_0,\delta_0,\hat\tau_0) = l(\delta_0,\delta_0,\hat\tau_0)
\end{equation}
\end{proof}

Thus, if the variational approximation $F$ favors $\glhat,\gfhat, \hat\tau$
over the rank 0 model, then it is guaranteed that the likelihood would also favor $\glhat,\gfhat,\hat\tau$
over the rank 0 model.
In other words, compared with the likelihood, the variational approximation is conservative in terms of preferring the rank 1 model to the rank 0 model.
This conservatism is a double-edged sword. On the one hand it means
that if the variational approximation finds structure it should be taken seriously. On the other hand
it means that the variational approximation could miss subtle structure. 

In practice Algorithm \ref{alg:r1b} can converge to a local optimum of $F$ that
is not as high as the trivial (rank 0) solution, $F(\delta_0,\delta_0,\delta_0,\hat\tau_0)$. 
We can add a check for this at the end of Algorithm
\ref{alg:r1b}, and set $\glhat=\gfhat=\delta_0$ and $\hat\tau = \hat\tau_0$ when this occurs.

\subsection{Identifiability} \label{sec:identifiability}

In EBMF each loading and factor is identifiable, at best, only up to a multiplicative constant
(provided $\G$ is a scale family). Specifically,
scaling the prior distributions $\gfk$ and $\glk$ by $c_k$ and $1/c_k$ respectively results in the same marginal likelihood, and also results in a corresponding scaling of the posterior distribution on the factors $\f_k$ and loadings $\l_k$ (e.g.~ it scales the posterior first moments by $c_k, 1/c_k$ and the second moments by $c_k^2,1/c_k^2$). However, this non-identifiability is not generally a problem, and if necessary it could be dealt with by re-scaling factor estimates to have norm 1.


\section{Software Implementation: \flash{}}
\label{sec:implementation}

We have implemented Algorithms \ref{alg:r1b}, \ref{alg:greedy} and \ref{alg:backfit} in an R
package, \flash{} (``factors and loadings
via adaptive shrinkage"). These algorithms
can fit the EBMF model for any choice of distributional family $\Gl,\Gf$: the user must simply provide a function to solve the EBNM problem for these prior families. 



One source of functions for solving the EBNM problem is the ``adaptive shrinkage" ({\tt ashr}) package, which implements methods from \cite{stephens:2017}. These methods solve the EBNM problem for several flexible choices of $\G$, including:
\begin{itemize}
\item $\G = \SN$, the set of all scale mixtures of zero-centered normals;
\item $\G = \SU$, the set of all symmetric unimodal distributions, with mode at 0;
\item $\G = \U$, the set of all unimodal distributions, with mode at 0;
\item $\G = \U_+$, the set of all non-negative unimodal distributions, with mode at 0.
\end{itemize}
These methods are computationally stable and efficient, being based on convex optimization methods \citep{koenker2014convex}
and analytic Bayesian posterior computations.
 
We have also implemented functions to solve the
EBNM problem for additional choices of $\G$
in the package {\tt ebnm} (\url{https://github.com/stephenslab/ebnm}). These
include $\G$ being the ``point-normal" family:
\begin{itemize}
    \item $\G = \PN$, the set of all distributions that are a mixture of a point mass at zero and a normal with mean 0.
\end{itemize} 
This choice is less flexible than those in {\tt ashr}, and involves non-convex optimizations, but can be faster.

Although in this paper we focus our examples
on sparsity-inducing priors with $\Gl=\Gf=\G$ we note that our software makes it easy to experiment
with different choices, some of which 
represent novel methodologies. For example, setting $\Gl=\U_{+}$
and $\Gf = \SN$ yields an EB version 
of semi-non-negative matrix factorization \citep{ding2008convex}, and we are aware of no
existing EB implementations for this problem.
Exploring the relative merits of these many possible options in different types of application will be an interesting direction for future work.

\subsection{Missing Data} \label{sec:missing}

If some elements of $Y$ are missing, then this is easily dealt with. For example, the sums over $j$ in \eqref{eqn:lhatfn} and \eqref{eqn:slfn} are simply computed
using only the $j$ for which $Y_{ij}$ is not missing.
This corresponds to an assumption that the missing
elements of $Y$ are ``missing at random'' \citep{rubin1976inference}.
In practice we implement this by setting $\precij=0$ whenever $Y_{ij}$ is missing
(and filling in the missing entries of $Y$ to an arbitrary number). This allows the implementation to exploit standard fast matrix multiplication routines, which cannot handle missing data.
If many data points are missing then it may be helpful
to exploit sparse matrix routines.

\subsection{Initialization} \label{sec:init}

Both Algorithms \ref{alg:r1b} and \ref{alg:greedy} require 
a rank 1 initialization procedure, {\tt init}.
Here, we use the {softImpute} function from the
package {\tt softImpute} \citep{mazumder2010spectral}, with penalty parameter $\lambda=0$, which essentially performs SVD when $Y$ is completely observed, but can also deal with missing values in $Y$.

The backfitting algorithm (Algorithm \ref{alg:backfit})
also requires initialization. One option is to use the greedy algorithm to initialize, which we call ``greedy+backfitting''.

\section{Numerical Comparisons}

We now compare our methods 
with several competing approaches.
To keep these comparisons manageable in scope
we focus attention on 
methods that aim to capture possible sparsity in $L$ and/or $F$.
For EBMF we present
results for two different shrinkage-oriented
prior families, $\G$: the scale mixture of normals ($\G = \SN$), and the point-normal family ($\G=\PN$). We denote these \flash{} and \flash{}\_pn respectively when we need to distinguish.
In addition we consider
Sparse Factor Analysis (SFA) \citep{Engelhardt2010}, SFAmix \citep{Gao2013}, Nonparametric Bayesian Sparse Factor Analysis (NBSFA) \citep{Knowles2011}, Penalized Matrix Decomposition \citep{Witten2009} (PMD, implemented in the R package {\tt PMA}), and Sparse SVD \citep{yang2014sparse} (SSVD, implemented in R package {\tt ssvd}). 
Although the methods we compare against involve
only a small fraction of the very large number of methods
for this problem, the methods were chosen 
to represent a wide range of
different approaches to inducing sparsity: SFA, SFAmix and NBSFA are three Bayesian approaches with quite different approaches to prior specification; PMD is based on a penalized likelihood with $L_1$ penalty on factors and/or loadings; and SSVD is based on iterative thresholding of singular vectors. We also compare with softImpute \citep{mazumder2010spectral}, which
does not explicitly model sparsity in $L$ and $F$,
but fits a regularized low-rank matrix using a nuclear-norm penalty. Finally, for reference
we also use standard (truncated) SVD. 

All of the Bayesian methods (\flash{}, SFA, SFAmix and NBSFA) are ``self-tuning'', at least to some extent, and we applied them here with default values.
According to \cite{yang2014sparse} SSVD is robust to choice of tuning parameters,
so we also ran SSVD with its default values, using the robust option ({\tt method="method"}). The softImpute method has a single
tuning parameter ($\lambda$, which controls the nuclear norm penalty), 
and we chose this penalty by orthogonal cross-validation 
(OCV; Appendix \ref{OCValgorithm}). The PMD method can use two tuning
parameters (one for $\l$ and one for $\f$) to allow different sparsity levels in $\l$ vs $\f$. However, since tuning two parameters can be inconvenient it also has the option to use a single parameter for both $\l$ and $\f$. We used OCV to tune parameters
in both cases, referring to the methods as PMD.cv2 (2 tuning parameters) and PMD.cv1 (1 tuning parameter). 






\subsection{Simple Simulations}

\subsubsection{A Single Factor Example} 

We simulated data with $n = 200, p = 300$ under the
single-factor model \eqref{eqn:EBMF} with sparse loadings, and a non-sparse factor:
\begin{align}
l_i &\sim \pi_0 \delta_0 + (1-\pi_0) \sum_{m = 1}^5 \frac{1}{5}N(0,\sigma_m^2) \\
f_j &\sim N(0,1) 
\end{align}
where $\delta_0$ denotes a point mass on 0, and $(\sigma^2_1,\dots,\sigma^2_5):=(0.25,0.5,1,2,4)$.
We simulated using three different levels of sparsity 
on the loadings, using $\pi_0 = 0.9,0.3,0$. (We set 
the noise precision $\tau = 1,1/16,1/25$ in these three cases
to make each problem not too easy and not too hard.) 

We applied all methods to this rank-1 problem, specifying the true value $K = 1$. (The NBSFA software does not provide the option to fix $K$, so is omitted here.)
We compare methods in their accuracy in estimating the true low-rank structure ($B := \l \f^T$)
using relative root mean squared error:
\begin{equation} \label{eqn:RRMSE}
\text{RRMSE}(\hat{B},B):= \sqrt{\frac{\sum_{i,j} (\hat{B}_{ij}-B_{ij})^2}{\sum_{i,j} B_{ij}^2}}.
\end{equation}



\begin{figure}[tb]
  \centering 
  {\includegraphics[width=15cm]{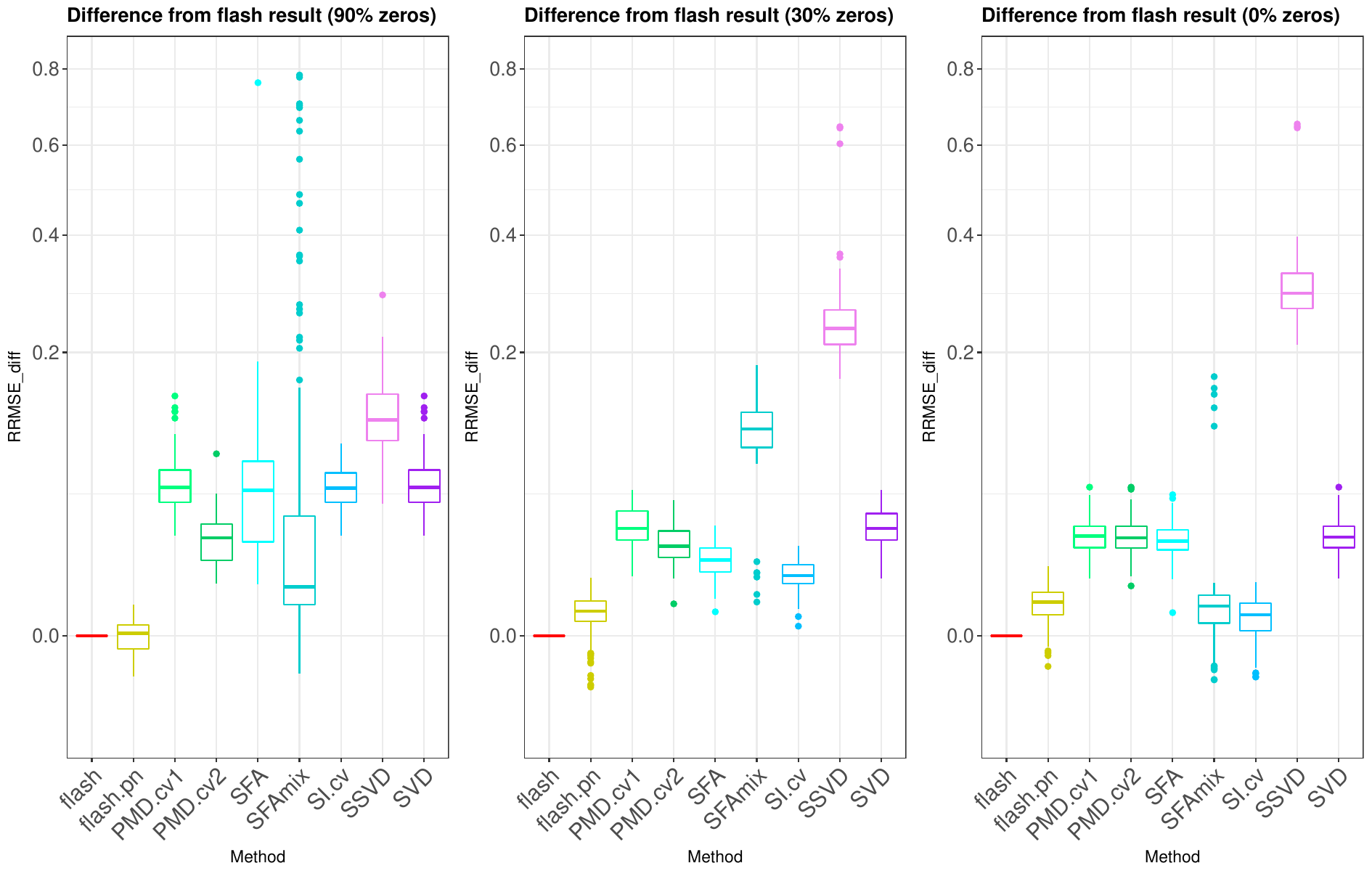}} 
   \caption{Boxplots comparing accuracy of \flash{} with several other methods in a simple rank-1 simulation.
 This simulation involves a single dense factor, and a loading that varies from strong sparsity ($90\%$ zeros, left) to no sparsity (right). Accuracy is measured by difference in each methods RRMSE from the \flash{} RRMSE, with smaller values indicating highest accuracy. The $y$ axis is plotted on a non-linear (square-root) scale to avoid the plots being dominated by poorer-performing methods.}
   \label{fig:r1}
\end{figure}

Despite the simplicity of this simulation, the methods vary greatly in performance
(Figure \ref{fig:r1}).
Both versions of \flash{} consistently outperform all the other methods across all
scenarios (although softImpute performs similarly in the non-sparse case).
The next best performances come from softImpute (SI.cv), PMD.cv2 and SFA, whose
relative performances depend on the scenario. All three consistently improve on,
or do no worse than, SVD. PMD.cv1 performs similarly to SVD.
The SFAmix method performs very variably, sometimes providing
very poor estimates, possibly due to poor convergence of the MCMC algorithm (it is the only method here that uses MCMC). The SSVD method consistently performs worse than 
simple SVD, possibly because it is more adapted to both factors and loadings being sparse
(and possibly because, following \citealt{yang2014sparse}, we did not use CV to tune its parameters).
Inspection of individual results suggests that the poor performance
of both SFAmix and SSVD is often due
to over-shrinking of non-zero loadings to zero.

\subsubsection{A Sparse Bi-cluster Example (Rank 3)} 

An important feature of our EBMF methods is that they estimate separate distributions $\gl,\gf$ for each factor and each loading, allowing them to adapt to any combination of sparsity in the factors and loadings. This 
flexibility is not easy to achieve in other ways.
For example, methods that use CV are generally limited to one or two tuning parameters because of the computational difficulties of searching over a larger space. 

To illustrate this flexibility we simulated data under the factor model \eqref{eqn:EBMF} with $n = 150, p=240$, $K=3, \tau=1/4$, and:
\begin{eqnarray}
\l_{1,i} & \sim N(0,2^2) \quad i=1,\dots,10 \\
\l_{2,i} & \sim N(0,1) \quad i=11,\dots,60 \\
\l_{3,i} & \sim N(0,1/2^2) \quad i = 61,\dots,150 \\
\f_{1,j} & \sim N(0,1/2^2) \quad j=1,\dots,80 \\
\f_{2,j} & \sim N(0,1) \quad j= 81,\dots,160 \\
\f_{3,j} & \sim N(0,2^2) \quad j= 161,\dots,240, 
\end{eqnarray}
with all other elements of $\l_k$ and $\f_k$ set to zero
for $k=1,2,3$.
This example has a sparse bi-cluster structure where  distinct groups of samples are each loaded on only one factor (Figure \ref{fig:K3}a), and both the size of the groups and number of variables in each factor vary. 

\begin{figure}[tb]
  \begin{subfigure}[b]{\textwidth}
  \centering       \includegraphics[width=10cm]{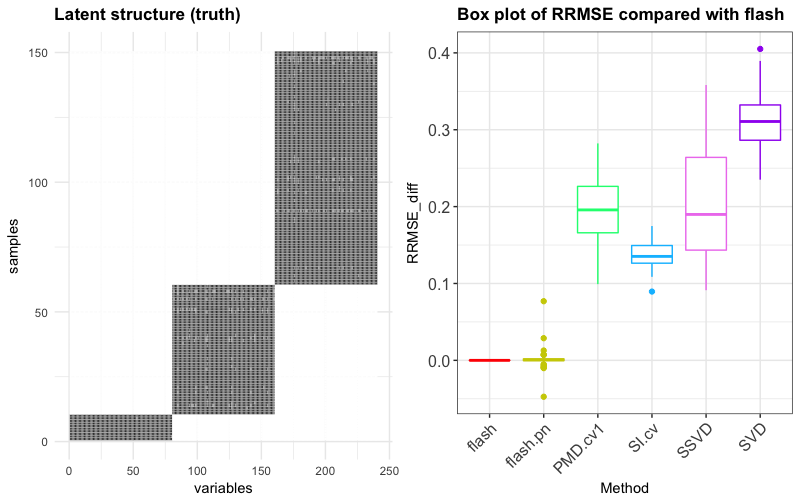}
   \caption{{\bf Left:} Illustration of the true latent rank-3 block structure used in these simulations. {\bf Right} boxplots comparing accuracy of \flash{} with several other methods across 100 replicates. Accuracy is measured by the difference of each methods RRMSE from the \flash{} RRMSE, so smaller is better.}
   \end{subfigure}
   
   \medskip
   
  \begin{subfigure}[b]{\textwidth}
    \centering \includegraphics[width=10cm]{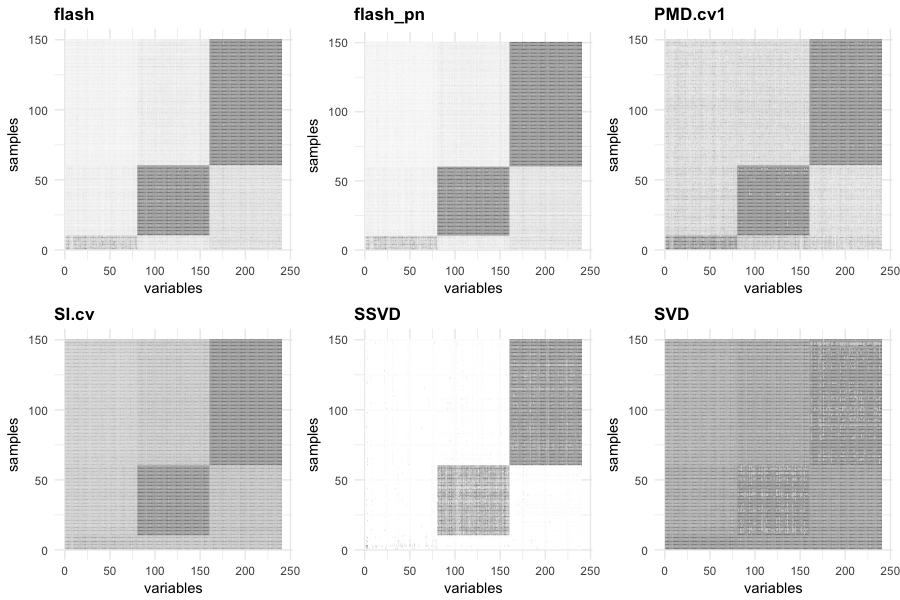} 
   \caption{Illustration of tendency of each method to either over-shrink the signal (SSVD) or under-shrink the noise (SI.cv, PMD.cv1, SVD) compared with \flash{}. Each panel shows the mean absolute value of the estimated structure from each method.}
   \end{subfigure}
   \caption{Results from simulations with sparse bi-cluster structure ($K=3$).}
\label{fig:K3}
\end{figure}

We applied \flash{}, softImpute, SSVD and PMD to this example. (We excluded
SFA and SFAmix since these methods do not model sparsity in both factors and loadings.)
The results (Figure \ref{fig:K3}) show that again
\flash{} consistently outperforms the other methods, and again the next best is softImpute.
On this example both SSVD and PMD outperform SVD. Although SSVD and PMD
perform similarly on average, their qualitative behavior is different: PMD insufficiently shrink the 0 values,
whereas SSVD shrinks the 0 values well but  overshrinks some of the signal, essentially
removing the smallest of the three loading/factor combinations (Figure \ref{fig:K3}b).



\subsection{Missing Data Imputation for Real Data Sets}

Here we compare methods
in their ability to impute missing data using five real data sets. In each case we ``hold out" (mask) some of the data points, and then apply the methods to obtain estimates of the missing values. 
The data sets are as follows:

{\it  \smallskip MovieLens 100K data,}
an (incomplete) $943 \times 1682$ matrix of
user-movie ratings (integers from 1 to 5)
\citep{harper2016movielens}.
Most users do not rate most movies, so the matrix
is sparsely observed (94\% missing), and contains 
about 100K observed ratings. We hold out a fraction of the observed entries and assess accuracy of methods in estimating these.
We centered and scaled the ratings for each user before analysis.

{\it  \smallskip GTEx eQTL summary data,}
a $16\,069 \times 44$ matrix of $Z$ scores computed
testing association of genetic variants (rows)
with gene expression in different human tissues (columns). These data come from the Genotype Tissue Expression (GTEx) project \citep{gtex2015genotype}, which 
assessed the effects of thousands of ``eQTLs" 
across 44 human tissues. (An eQTL is a genetic variant that is associated with expression of a gene.)
To identify eQTLs, the project tested for association between expression and every near-by genetic variant, each test yielding a $Z$ score. The data used here are the $Z$ scores for the most significant genetic variant for each gene (the ``top'' eQTL). See Section \ref{GTExdata} for more detailed analyses of these data.

{\it  \smallskip Brain Tumor data,} a $43 \times 356$ matrix of gene expression measurements
 on 4 different types of brain tumor  \citep[included in
the {\tt denoiseR} package,][]{josse2018denoiser}.
 We centered each column before analysis.
 
{\it  \smallskip Presidential address data,}
a $13 \times 836$ matrix of word counts
from the inaugural addresses of 13 US presidents (1940--2009) \citep[also included in
the {\tt denoiseR} package,][]{josse2018denoiser}.
Since both row and column means 
vary greatly we centered and scaled both rows and columns before analysis, using the {\tt biScale} function from {\tt softImpute}.



{\it \smallskip Breast cancer data,} a $251 \times 226$ matrix of gene expression measurements from \cite{Carvalho2008}, which were used as an example in the paper introducing NBSFA \citep{Knowles2011}. Following
\cite{Knowles2011} we centered each column (gene) before analysis.


\begin{figure}[tb]
  \centering 
  {\includegraphics[width=13cm]{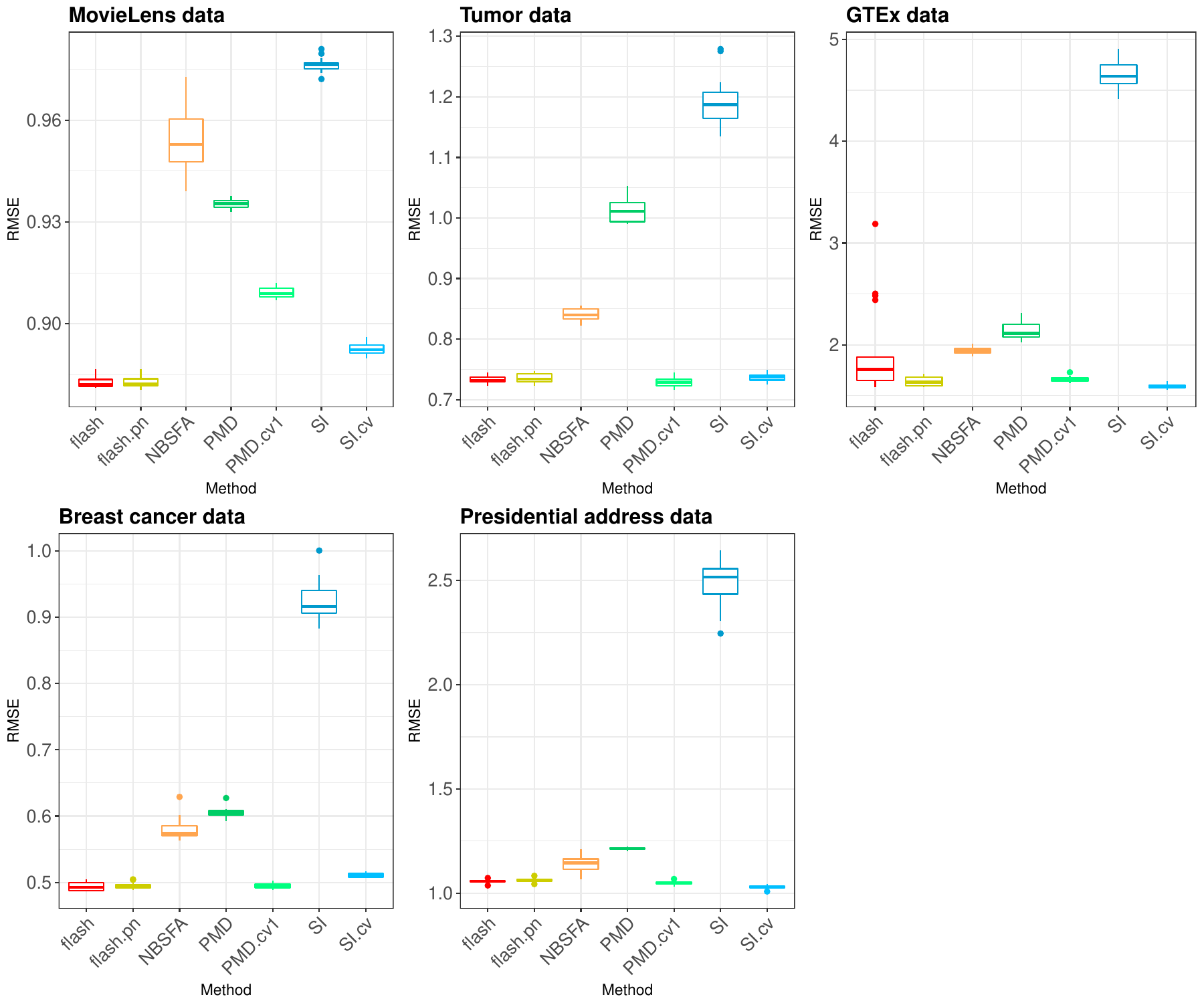}}
   \caption{Comparison of the accuracy of different methods in imputing missing data. Each panel shows a boxplot of error rates (RMSE) for 20 simulations based on masking observed entries in a real data set. }
\label{fig_4}
\end{figure}


\bigskip

Among the methods considered above, only
\flash{}, PMD and softImpute can handle missing data. We add
NBSFA \citep{Knowles2011} to these comparisons. 
To emphasize the importance of parameter tuning
we include results for PMD and softImpute with default settings (denoted PMD, SI) as well as using cross-validation (PMD.cv1, SI.cv). 

For these real data the appropriate value of $K$ is, of course, unknown. Both \flash{} and NBSFA automatically estimate $K$. For PMD and softImpute 
we specified $K$ based on the values inferred by \flash{} and NBSFA. (Specifically, we used $K=10,30,20,10,40$ respectively for the five data sets.)

We applied each method
to all 5 data sets, using 10-fold OCV (Appendix \ref{OCValgorithm}) to mask data points for imputation, repeated 20 times (with different random number seeds) for each data set. We measure imputation accuracy 
using root mean squared error (RMSE): 
\begin{equation}
\label{missRMSE}
\text{RMSE}(\hat{Y},Y; \mathit{\Omega}) = \sqrt{\frac{1}{|\mathit{\Omega|}}\sum_{ij \in \mathit{\Omega}}(Y_{ij} - \hat{Y}_{ij})^2}.
\end{equation}
where $\mathit{\Omega}$ is the set of indices of the held-out data points.

The results are shown in Figure \ref{fig_4}. Although the ranking of methods varies among data sets, \flash{}, PMD.cv1 and SI.cv perform similarly on average, and consistently outperform NBSFA, which in turn typically outperforms (untuned) PMD and unpenalized softImpute. 
These results highlight the importance of appropriate tuning for the penalized methods, and also the effectiveness of the EB method in \flash{} to provide
automatic tuning. 

In these comparisons, as in the simulations, the two \flash{} methods typically performed similarly. The exception is the GTEx data, where the scale mixture of normals ($\G=\SN$) performed worse. Detailed investigation revealed this
to be due to a very small number of very large ``outlier'' imputed values, well outside the range of the observed data, which grossly inflated RMSE. These outliers were so extreme that it should be possible to implement a filter to avoid them. However, 
we did not do this here
as it seems useful to highlight this unexpected behavior.
(Note that this occurs only when data are missing, and even then only in one of the five data sets considered here.)

\subsection{Sharing of Genetic Effects on Gene Expression Among Tissues}
\label{GTExdata}

To illustrate \flash{} in a scientific application, we applied it to the GTEx data
described above, a $16,069 \times 44$ matrix of $Z$ scores, with $Z_{ij}$ reflecting the strength (and direction) of effect of eQTL $i$ in tissue $j$.
We applied \flash{} with $\G=\SN$
using the greedy+backfitting algorithm (i.e.~the backfitting algorithm, initialized using the greedy algorithm). 


 

The \flash{} results yielded 26 factors (Figure \ref{gtexloading1}-\ref{gtexloading2}) which 
summarize the main patterns of eQTL
sharing among tissues (and, conversely, the main  patterns of tissue-specificity).
For example, the first factor has approximately equal weight for every tissue, and reflects
the fact that many eQTLs show similar effects across all 44 tissues. The second factor
has strong effects only in the 10 brain tissues, from which we infer that some eQTLs show much stronger effects in brain tissues than other tissues. 

Subsequent factors tend to be sparser, and many have a strong effect in only one tissue, capturing ``tissue-specific" effects. For example, the 3rd factor shows a strong effect only in whole blood, and captures eQTLs that have much stronger effects in whole blood than other tissues. (Two tissues, ``Lung" and ``Spleen", show very small effects in this factor but with the same sign as blood. This is intriguing since the lung has recently been found to make blood cells---see \citealt{Lefrancais2017}---and a key role of the spleen is storing of blood cells.)
Similarly Factors 7, 11 and 14 capture effects specific to ``Testis", ``Thyroid" and ``Esophagus Mucosa" respectively.

A few other factors show strong effects in a small number of tissues that are known to be
biologically related, providing support that the factors identified are scientifically meaningful. 
For example, factor 10 captures the two tissues related to the cerebellum, ``Brain Cerebellar Hemisphere" and ``Brain Cerebellum".
Factor 19 captures tissues related to female reproduction, ``Ovary", ``Uterus" and ``Vagina". Factor 5 captures ``Muscle Skeletal", with small but concordant effects in the heart tissues (``Heart Atrial Appendage" and ``Heart Left Ventricle"). Factor 4, captures the two skin tissues 
(``Skin Not Sun Exposed Suprapubic", ``Skin Sun Exposed Lower leg") and also ``Esophagus Mucosa", possibly reflecting the sharing of squamous cells that are found in both the surface of the skin, and the lining of the digestive tract. In factor 24, ``Colon Transverse" and ``Small Intestine Terminal Ileum" show the strongest effects (and with same sign), reflecting some sharing of effects in these intestinal tissues.
Among the 26 factors, only a few are difficult to interpret biologically (e.g.~ factor 8).

To highlight the benefits of sparsity,
we contrast the \flash{} results with
those for softImpute, which was the best-performing method in the
missing data assessments on these data,
but which uses a nuclear norm penalty that
does not explicitly reward sparse factors or loadings. The first eight softImpute factors are shown in Figure \ref{fig:GTEXSI}. The softImpute results---except for the first two factors---show little resemblance to the flash{} results, and in our view are harder to interpret.



\begin{figure}[hp]
  \centering 
  {\includegraphics[width=14cm]{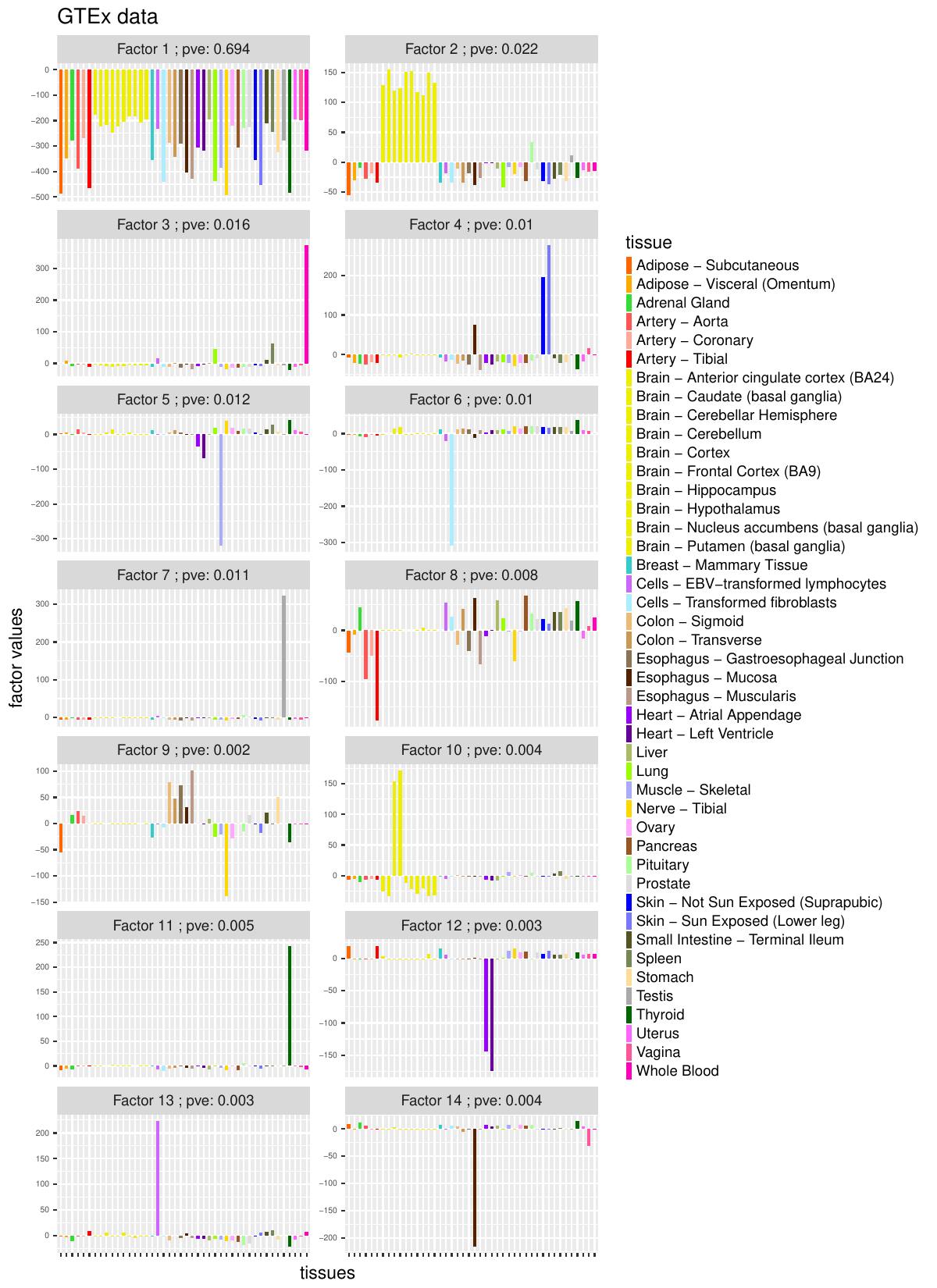}}
   \caption{Results from running \flash{} on GTEx data (factors 1 - 8). The pve ("Percentage Variance Explained") for loading/factor $k$ is defined as $\text{pve}_k := s_k / (\sum_k s_k+ \sum_{ij} 1/\tau_{ij})$ where $s_k:=\sum_{ij} (\ll_{ki} \ff_{kj})^2$. It is
a measure of the amount of signal in the data captured by loading/factor $k$ (but its naming as "percentage variance explained" should be considered loose since the factors are not orthogonal).}
   \label{gtexloading1}
\end{figure}

\begin{figure}[hp]
  \centering 
  {\includegraphics[width=14cm]{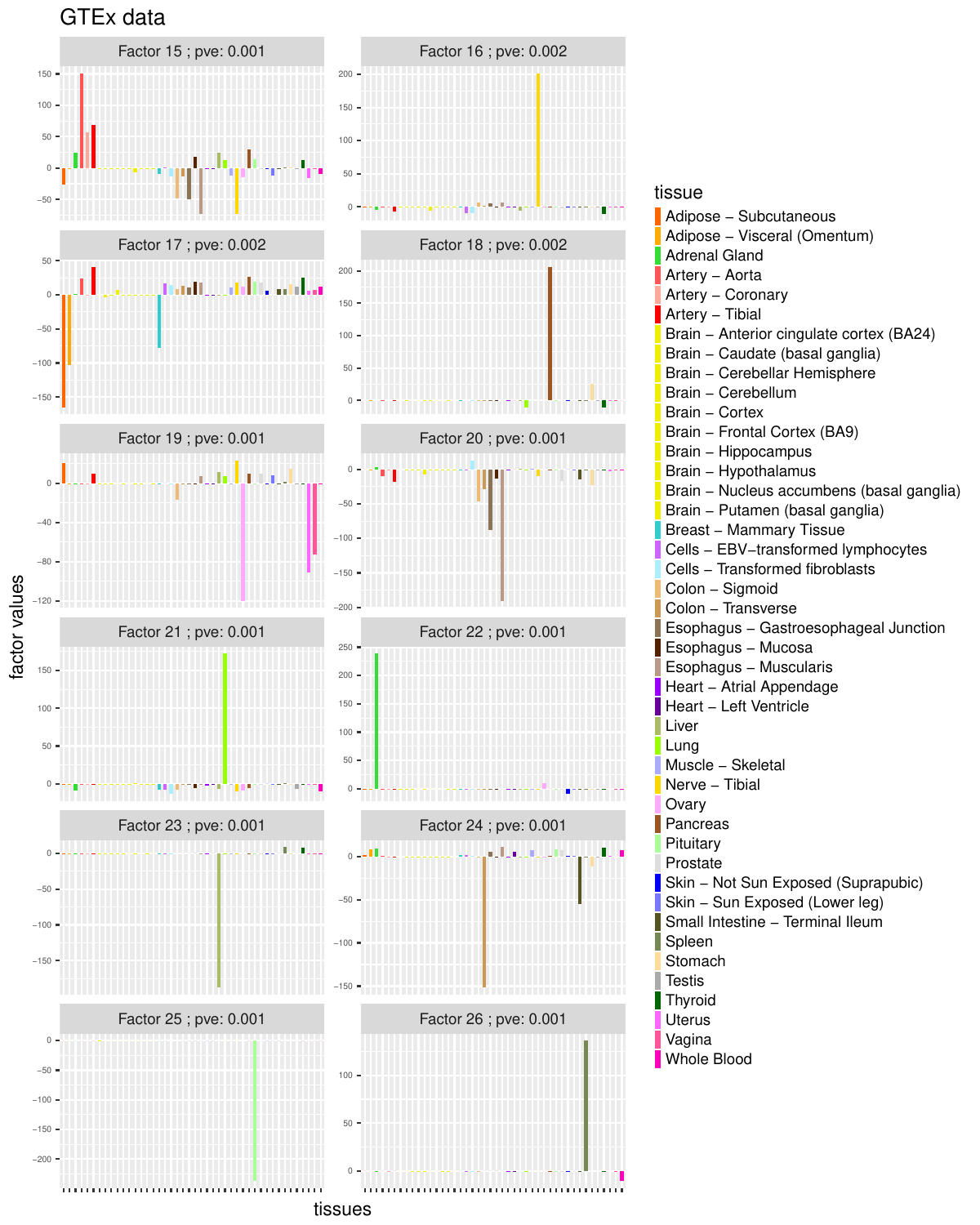}}
   \caption{Results from running \flash{} on GTEx data (factors 15 - 26)}
   \label{gtexloading2}
\end{figure}
 
\begin{figure}[hp]
  \centering 
  {\includegraphics[width=15cm]{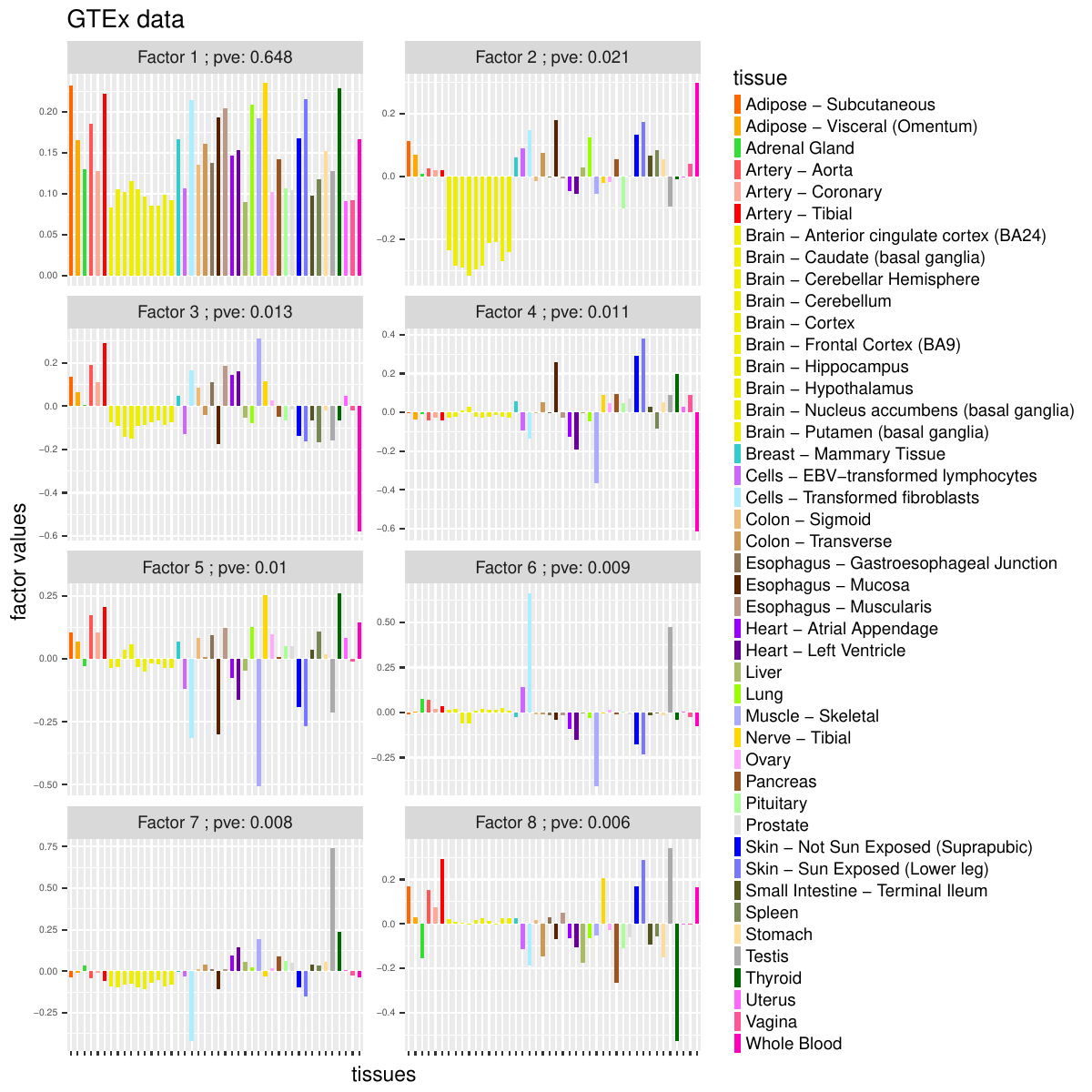}} 
   \caption{Results from running softImpute on GTEx data (factors 1-8). The factors are both less sparse and less interpretable than the \flash{} results.}
   \label{fig:GTEXSI}
\end{figure}

\subsection{Computational Demands}

It is difficult to make general statements
about computational demands of our methods, because both the number of factors and number of iterations per factor can vary considerably depending on the data. However, to give a specific example, running our current implementation of the greedy algorithm on the GTEx data (a 16,000 by 44 matrix) takes about 140s (wall time) 
for $\G=\PN$ and 650s for
$\G=\SN$
(on a 2015 {\sc MacBook Air} with a 2.2 GHz Intel Core i7 processor and 8Gb RAM).
By comparison, a single run of softImpute
without CV takes 2-3s, so a naive
implementation of 5-fold CV with 10 different tuning parameters and 10 different values of $K$ would take over 1000s (although one could
improve on this by use of warm starts for example).

\section{Discussion}

Here we discuss some potential extensions or modifications of our work.

\subsection{Orthogonality Constraint}
Our formulation here does not require the factors or loadings to be orthogonal. In scientific applications
we do not see any particular reason to expect
underlying factors to be orthogonal. However, imposing such a constraint could have computational or mathematical advantages. Formally adding such a constraint to our objective function seems tricky, but
it would be straightforward to modify our algorithms to include an orthogonalization step each update.
This would effectively result in an EB version of the SSVD algorithms in \cite{yang2014sparse}, and it seems
likely to be computationally faster than our current approach. One disadvantage of this approach is that it is unclear what optimization problem such an algorithm would solve (but the same is true of SSVD, 
and our algorithms have the advantage 
that they deal with missing data.)

\subsection{Non-negative Matrix Factorization}

We focused here on the potential for
EBMF to induce sparsity on loadings and factors. However, EBMF can also encode other assumptions.
For example, to assume the loadings and factors are non-negative, simply restrict $\G$ to be a family of non-negative-valued distributions,
yielding ``Empirical Bayes non-negative Matrix Factorization" (EBNMF). 
Indeed, the {\tt ashr} software
can already solve the EBNM problem for some such families $\G$, and so \flash{} already implements EBNMF. In preliminary assessments we found that the greedy approach is problematic here: the non-negative constraint makes it harder for later factors to compensate for errors in earlier factors.  However, it is straightforward to apply the backfitting algorithm to fit EBNMF, with 
initialization by any existing
NMF method. The performance of this approach
is an area for future investigation.

\subsection{Tensor Factorization}

It is also straightforward to extend EBMF
to tensor factorization, specifically a CANDECOMP/PARAFAC
decomposition \citep{kolda2009tensor}:
\begin{align}
Y_{ijm}  &= \sum_{k =1}^K l_{ki} f_{kj} h_{km} + E_{ijm} \\
l_{k1},\dots,l_{kn} &\iid \glk, \quad \glk \in \G \\
f_{k1},\dots,f_{kp} &\iid \gfk, \quad \gfk \in \G \\
h_{k1},\dots,h_{kr} &\iid \ghk, \quad \ghk \in \G \\
E_{ijm} &\iid N(0,1/\tau_{ijm}).
\end{align}
The variational approach is easily extended to this
case \citep[a generalization of methods in ][]{hore2016tensor}, and updates that
increase the objective function can be constructed by solving an EBNM problem, similar to EBMF. 
It seems likely that issues of convergence to local optima, and the need for good initializations, will need some attention to obtain good practical performance. However, results in \cite{hore2016tensor} are promising, and the automatic-tuning feature
of EB methods seems particularly attractive here. For example, extending PMD to this case---allowing for 
different sparsity levels in $l,f$ and $h$---would require 3 penalty parameters even in the rank 1 case,
making it difficult to tune by CV.

\subsection{Non-Gaussian Errors} \label{sec:nongauss}

It is also possible
to extend the variational approximations used
here to fit non-Gaussian models, such as binomial data; see for example
\cite{Jordan2000,seeger2012fast,klami2015polya}. The extension of our EB methods using these ideas is detailed in \cite{wang2017}.


\section*{Acknowledgements}

We thank P.~Carbonetto for computational assistance, and P.~Carbonetto, D.~Gerard, and A.~Sarkar for helpful conversations and comments on a draft manuscript. 
Computing resources were provided by the University of Chicago Research Computing Center. This work was supported by NIH grant HG002585 and by a grant from the Gordon and Betty Moore Foundation (Grant GBMF \#4559). 


\newpage

\appendix

\section{Variational EBMF with \texorpdfstring{$K$}{K} Factors} \label{app:ebfa_k}

Here we describe in detail the variational approach to the $K$ factor model, including deriving updates that
we use to optimize the variational objective. (These
derivations naturally include the $K=1$ model as a special case, and our proof of Proposition \ref{prop:rK} below includes Proposition \ref{prop:EBNM} as a special case.)

Let $\ql,\qf$ denote the variational distributions on the $K$ loadings/factors:
\begin{align}
\ql(\l_1,\cdots,\l_K) & = \prod_k q_{\l_k}(\l_k) \\
\qf(\f_1,\dots,\f_K) & = \prod_k q_{\f_k}(\f_k).
\end{align}
The objective function $F$ \eqref{elbow} is thus a function of $\ql = (q_{\l_1},\dots,q_{\l_K})$, 
$\qf = (q_{\f_1},\dots,q_{\f_K})$, $\gl = (g_{\l_1},\dots,g_{\l_K})$
and $\gf =(g_{\f_1},\dots,g_{\f_K})$, as well as the precision $\prec$:
\begin{align} \label{eqn:elboK}
F(\ql,\qf,\gl,\gf,\prec) &= \int \prod_k q_{\l_k}(\l_k)q_{\f_k}(\f_k)\log \frac{p(Y,\l,\f;g_{\l_1},g_{\f_1},\cdots,g_{\l_K}, g_{\f_K},\prec)}{\prod_k q_{\l_k}(\l_k)q_{\f_k}(\f_k)} \, d\l_k \, d\f_k, \\
&= E_{\ql,\qf}\log p(Y|\l,\f;\prec) + 
\sum_k E_{q_{\l_k}} \log \frac{g_{\l_k}(\l_k)}{q_{\l_k}(\l_k)} + 
\sum_k E_{q_{\f_k}} \log \frac{g_{\f_k}(\f_k)}{q_{\f_k}(\f_k)}.
\end{align}

We optimize $F$ by iteratively updating parameters relating to $\prec$, a single loading $k$ $(q_{\l_k},\glk)$ or factor $k$ $(q_{\f_k},\gfk)$, keeping other parameters fixed. 
We simplify implementation by keeping track of only the first and second moments of the distributions $q_{\l_k}$ and $q_{\f_k}$, which we denote $\ll_k,\llsq_k,\ff_k,\ffsq_k$.
We now describe each kind of update in turn.

\subsection{Updates for Precision Parameters}
\label{sec:est_var}

Here we derive updates to optimize $F$ over the precision parameters $\prec$.
Focusing on the parts of
$F$ that depend on $\prec$ gives:
\begin{align}
F(\prec) & = E_{\ql} E_{\qf} \sum_{ij} 0.5 \log(\tau_{ij}) - 0.5 \tau_{ij} (Y_{ij} - \sum_k l_{ki} f_{kj})^2 + \text{const} \\
&= 0.5 \sum_{ij} \left[ \log(\tau_{ij}) + \tau_{ij} \ERsq_{ij} \right] + \text{const}
\end{align}
where $\ERsq$ is defined by:
\begin{align} \label{eqn:R2}
\ERsq_{ij} &:= E_{\ql,\qf}[(Y_{ij}-\sum_{k=1}^K \l_{ki} \f_{kj})^2] \\ 
&= (Y_{ij} - \sum_k \ll_{ki} \ff_{kj})^2 - \sum_k (\ll_{ki})^2 (\ff_{kj})^2 + \sum_k \llsq_{ki} \ffsq_{kj}.
\end{align}

If we constrain $\prec \in \T$ then we have
\begin{equation}
\hat{\prec} = \arg \max_{\prec \in \T} \sum_{ij} [\log(\tau_{ij}) - \tau_{ij} \ERsq_{ij}].
\end{equation}
For example, assuming constant precision
$\prec_{ij} = \prec$ yields:
\begin{equation}
\hat{\prec} =\frac{NP}{\sum_{ij} \ERsq_{ij}}.
\end{equation}
Assuming column-specific precisions ($\prec_{ij}=\prec_j$), which is the default in our software, yields:
\begin{equation}
\hat{\prec}_j =\frac{N}{\sum_{i} \ERsq_{ij}}.
\end{equation}
Other variance structures are considered in Appendix A.5 of \cite{wang2017}.

\subsection{Updating Loadings and Factors} \label{app:single_factor_updates_rankk}

The following Proposition, which generalizes Proposition \ref{prop:EBNM} in the main text, shows how updates for loadings (and factors) for the $K$-factor EBMF model can be achieve by solving an EBNM problem. 

\begin{proposition}  \label{prop:rK}
For the $K$-factor model, $\arg \max_{q_{\l_k}, \glk} F(\ql,\qf,\gl,\gf,\prec)$ is solved by solving an EBNM problem. Specifically 
\begin{equation}
\arg \max_{q_{\l_k}, \glk} F(\ql,\qf,\gl,\gf,\prec)  =  \ebnm (\hat{\l}(R^k,\ff_k,\ffsq_k,\prec),\bs_l(\ffsq_k,\prec))
\end{equation}
where the functions $\hat{\l}$ and $\bs_l$
are given by \eqref{eqn:lhatfn} and \eqref{eqn:slfn}, $\ff_k, \ffsq_k \in R^p$ denote the vectors whose elements are the first and second moments of $\f_k$ under $q_{\f_k}$, and $R^k$ denotes the residual matrix
\eqref{eqn:resid}.

Similarly, 
$\arg \max_{q_{\f_k}, \gfk}F(\ql,\qf,\gl,\gf,\prec)$
is solved by solving an EBNM problem. Specifically,
\begin{equation}
\arg \max_{q_{\f_k}, \gfk} F(\ql,\qf,\gl,\gf,\prec)  =  \ebnm (\hat{\f}(R^k,\ll_k,\llsq_k,\prec),\bs_f(\llsq_k,\prec))
\end{equation}
where 
the functions $\hat{\f}: R^{n \times p} \times R^n \times R^n \times R^{n \times p} \rightarrow R^p$ and $\bs_f:R^n \times R^{n \times p} \rightarrow R^p$ are given by
\begin{align} \label{eqn:fhatfn}
\hat{\f}(Y,\bm{v},\bm{w},\prec)_j &:=  \frac{\sum_i \precij Y_{ij} v_i}{\sum_i \precij w_i}, \\ \label{eqn:sffn}
\bm{s}_{\f}(\bm{w},\prec)_j &:= \left(\sum_i \precij w_i \right)^{-0.5}.
\end{align}

\end{proposition}


\subsubsection{A Lemma on the Normal Means Problem}

To prove Proposition \ref{prop:rK} we introduce
a lemma that characterizes the solution of the
normal means problem in terms of an objective that
is closely related to the variational objective.

Recall that the EBNM model is:
\begin{align} \label{eqn:normalmeans2}
\bx &= \btheta + \be \\
\theta_1,\dots,\theta_n & \iid g, \quad g \in \G.\label{eqn:normalmeans2_2}
\end{align}
where $e_i \sim N(0,s_i^2)$. 
 
Solving the EBNM problem involves estimating $g$ by maximum likelihood:
\begin{equation} \label{eqn:ghat2}
\hat{g} = \arg \max_{g \in \G} l(g),
\end{equation} 
where
\begin{equation}
l(g) = \log p(\bx | g).
\end{equation}
It also involves finding the posterior distributions:
\begin{equation} \label{eqn:post2}
p(\btheta | \bx, \hat{g}) = \prod_j p(\theta_j | \bx, \hat{g})  \propto \prod_j \hat{g}(\theta_j) p(x_j | \theta_j, s_j).
\end{equation}

\def\post{p_{\theta | x, g}}
\def\qt{q_\theta}
\def\FNM{F^{\text{NM}}}

\begin{lemma}
\label{Lemma_1}
Solving the EBNM problem also solves:
\begin{equation} \label{eqn:EBNM-Fversion}
\max_{q_\theta, g\in\G} \FNM(\qt, g)
\end{equation}
where 
\begin{equation} \label{eqn:EBNMF}
\FNM(\qt,g) = E_{\qt}\left[-\frac{1}{2} \sum_j (A_j \theta_j^2 -2 B_j \theta_j)\right] + E_{\qt}\log \frac{g(\btheta)}{\qt(\btheta)} + \text{const}
\end{equation} 
with $A_j  = 1/s_j^2$ and $B_j = x_j/s_j^2$, and $g(\btheta):=\prod_j g(\theta_j)$.

Equivalently, \eqref{eqn:EBNM-Fversion}-\eqref{eqn:EBNMF} is solved by 
$g=\hat{g}$ in \eqref{eqn:ghat2} and $\qt = p(\btheta | \bx, \hat{g})$ in \eqref{eqn:post2},
with $x_j = B_j/A_j$ and $s^2_j = 1/A_j$.
\end{lemma}

\begin{proof}
The log likelihood can be written as
\begin{align}
l(g)  &:= \log [p(\bx| g)] \\
& = \log[ p(\bx, \btheta | g) / p(\btheta | \bx,g)] \\
& = \int \qt(\btheta)  \log \frac{p(\bx, \btheta | g)}{p(\btheta | \bx,g)} d\btheta \\
& = \int \qt(\btheta)  \log \frac{p(\bx, \btheta | g)}{\qt(\btheta)} d\btheta + \int \qt(\btheta) \log \frac{\qt(\btheta)}{p(\btheta | \bx,g)} d\btheta \\
& = \FNM(\qt,g) + D_{KL}(\qt || \post) \label{eqn:lik-F-D}
\end{align}
where
\begin{equation}
\label{def_elbow}
\FNM(\qt,g) = \int \qt(\btheta) \log \frac{p(\bx, \btheta | g)}{\qt(\btheta)} \, d\btheta
\end{equation}
and 
\begin{equation}
\label{def_KL}
D_{KL}(\qt || \post) = -\int  \qt(\btheta) \log \frac{p(\btheta | \bx, g) }{\qt(\btheta)} \, d\btheta
\end{equation}
Here $\post$ denotes the posterior distribution $p(\btheta | \bx, g)$.
This identity holds for any distribution $\qt(\btheta)$.

Rearranging \eqref{eqn:lik-F-D} gives:
\begin{equation}
\FNM(\qt,g) = l(g) - D_{KL}(\qt || \post).
\end{equation}
Since $D_{KL}(\qt || \post) \geq 0$, with equality when $\qt=\post$, $\FNM(\qt,g)$ is maximized over $\qt$ by setting $\qt=\post$.
Further 
\begin{equation}
\label{F_l_link_1}
\max_{\qt} \FNM(\qt,g) = l(g),
\end{equation}
so
\begin{equation}
\arg \max_{g \in \G} \max_{\qt} \FNM(\qt,g) = \arg \max_{g \in \G} l(g) = \hat{g}.
\end{equation} 
It remains only to show that $\FNM$ has the form \eqref{eqn:EBNMF}.

By \eqref{eqn:normalmeans2} and \eqref{eqn:normalmeans2_2}, we have
\begin{equation}
\log p(\bx, \btheta | g) = -\frac{1}{2} \sum_j s_j^{-2} (x_j - \theta_j)^2 + \log g(\btheta) + \text{const}.
\end{equation}
Thus
\begin{equation}
\label{proof_Fnm}
 \FNM(\qt,g) = E_{\qt}\left[-\frac{1}{2} \sum_j (A_j \theta_j^2 -2 B_j \theta_j)\right] + \E_{\qt}\log \frac{g(\btheta)}{\qt(\btheta)}+ \text{const}.
\end{equation}
\end{proof}

\subsubsection{Proof of Proposition \ref{prop:rK}}

We are now ready to prove Proposition \ref{prop:rK}.
\begin{proof}
We prove the first part of the proposition since the proof for the second part is essentially the same.

The objective function \eqref{eqn:elboK} is:
\begin{align} 
F(\ql,\qf,\gl,\gf,\prec) &= E_{\ql,\qf}\log p(Y|\l,\f;\prec) + 
\sum_k E_{q_{\l_k}} \log \frac{g_{\l_k}(\l_k)}{q_{\l_k}(\l_k)} + 
\sum_k E_{q_{\f_k}} \log \frac{g_{\f_k}(\f_k)}{q_{\f_k}(\f_k)} \\
& = E_{q_{\l_k}}\left[-\frac{1}{2} \sum_i (A_{ik} l_{ki}^2 -2 B_{ik} l_{ki})\right]
 + E_{q_{\l_k}} \log \frac{\glk(\l_k)}{q_{\l_k}(\l_k)} + C_1 \label{marginalF_K}
\end{align}
where $C_1$ is a constant with respect to $q_{\l_k}, \glk$ and
\begin{align}
A_{ik} &=  \sum_j \tau_{ij} \Efksq\\
B_{ik} &=  \sum_j \tau_{ij}\left(R^k_{ij}  \Efk \right).
\end{align} 

Based on Lemma~\ref{Lemma_1}, we can solve this optimization problem \eqref{marginalF_K} by solving the EBNM problem with:
\begin{align}
x_i & = \frac{\sum_j \tau_{ij}\left(R^k_{ij}  \Efk \right)}{\sum_j \tau_{ij} \Efksq} \\
s^2_i &= \frac{1}{\sum_j \tau_{ij} \Efksq}.
\end{align}
\end{proof}

\subsection{Algorithms}

Just as with the rank $1$ EBMF model,
the updates for the rank $K$ model 
require only the first and second
moments of the variational distributions $q$. 
Thus we implement the updates in
algorithms that keep track of
the first moments ($\ll:=(\ll_1,\dots,\ll_K)$ and $\ff:=(\ff_1,\dots,\ff_K)$) and second moments
($\llsq:=(\llsq_1,\dots,\llsq_K)$ and $\ffsq:=(\ffsq_1,\dots,\ffsq_K)$), and the precision $\prec$.

Algorithm \ref{alg:update} implements a basic update for $\prec$, and for the parameters relating to a single factor $k$ ($\ll_k,\llsq_k,\ff_k,\ffsq_k$).
Note that the latter updates are identical to the updates for fitting the single factor EBMF model, but with $Y_{ij}$ replaced with the residuals obtained by removing the estimated effects of the other $k-1$ factors.

\begin{algorithm}[H]
\caption{Single-factor update for EBMF (rank K)} \label{alg:update}
\begin{algorithmic}[1]
\Require A data matrix $Y$ ($n \times p$)
\Require A function, ${\tt ebnm}(\bx,\bs) \rightarrow (\widebar{\btheta},\widebar{\btheta^2})$, that solves the EBNM problem \eqref{eqn:normalmeans}-\eqref{eqn:theta} and returns the first and second posterior moments \eqref{eqn:postmean}-\eqref{eqn:postmean2}.
\Require Current values for first moments $\ll:=(\ll_1,\dots,\ll_K)$ and $\ff:=(\ff_1,\dots,\ff_K)$.
\Require Current values for second moments $\llsq:=(\llsq_1,\dots,\llsq_K)$ and $\ffsq:=(\ffsq_1,\dots,\ffsq_K)$.
\Require An index $k$ indicating which loading/factor to compute updated values for.
\State Compute matrix of expected squared residuals, $\ERsq$, using \eqref{eqn:R2}
\State $\tau_j \gets n/\sum_i \ERsq_{ij}$. [Assumes column-specific variances; can be modified to make other assumptions.]
\State Compute residual matrix $R^k:= Y-\sum_{k' \neq k} \ll_{k'} \ff^T_{k'}$.
	\State Compute $\hat{\l}(R^k,\ff_k,\ffsq_k,\prec)$ and its standard error $\bs_l(\prec,\ffsq_k)$, using \eqref{eqn:lhatfn} and \eqref{eqn:slfn}.
	\State $(\ll_k,\llsq_k) \gets {\tt ebnm}(\hat{\l},\bs_l)$.   
	\State Compute $\hat{\f}(R^k,\ll_k,\llsq_k,\prec)$ and its standard error $\bs_f(\llsq,\prec_k)$. 
	\State $(\ff_k,\ffsq_k) \gets {\tt ebnm}(\hat{\f},\bs_f)$. \\
\Return updated values $\ll_k,\llsq_k,\ff_k,\ffsq_k,\prec$.
\end{algorithmic}
\end{algorithm}

Based on these basic updates we implemented two algorithms for fitting the $K$-factor EBMF model: the greedy algorithm,
and the backfitting algorithm, as follows.

\subsubsection{Greedy Algorithm} \label{app:greedy}


The greedy algorithm is a forward procedure that, at the $k$th step, adds new factors and loadings $\l_k$, $\f_k$ by optimizing over $q_{\l_k},q_{\f_k},\glk,\gfk$ while keeping 
the distributions related to previous factors fixed.
Essentially this involves fitting
the single-factor model to the residuals obtained
by removing previous factors.
The procedure stops adding
factors when the estimated new factors (or loadings) are identically zero. The algorithm as follows:

\begin{algorithm}[H]
\caption{Greedy Algorithm for EBMF} \label{alg:greedy}
\begin{algorithmic}[1]
\Require A data matrix $Y$ ($n \times p$)
\Require A function, ${\tt ebnm}(\bx,\bs) \rightarrow (\widebar{\btheta},\widebar{\btheta^2})$, that solves the EBNM problem \eqref{eqn:normalmeans}-\eqref{eqn:theta} and returns the first and second posterior moments \eqref{eqn:postmean}-\eqref{eqn:postmean2}.
\Require A function, ${\tt init}(Y) \rightarrow (\l;\f)$ that
provides initial estimates for the loadings and factors (see Section \ref{sec:init}).
\Require A function ${\tt single\_update}(Y,\ll,\ff,\llsq,\ffsq,k) \rightarrow(\ll_k,\llsq_k,\ff_k,\ffsq_k,\prec)$ implementing Algorithm \ref{alg:update}.
\State initialize $K \gets 0$.
\Repeat
\State $K \gets K+1$.
\State Compute residual matrix $R_{ij} = Y_{ij} - \sum_{k=1}^{K-1} \ll_{ki} \ff_{kj}$.
\State Initialize first moments $(\ll_K,\ff_K) \gets {\tt init}(R)$.
\State Initialize second moments by squaring first moments: 
$\llsq_K \gets \ll_K^2; \ffsq_K \gets \ff_K^2$. 
\Repeat 
	\State $(\ll_K,\llsq_K,\ff_K,\ffsq_K,\prec) \gets
    {\tt single\_update}(Y,\ll,\llsq,\ff,\ffsq,K)$
\Until converged
\Until $\ff_K$ is identically 0 or $\ll_K$ is identically 0. \\
\Return $\ll,\llsq,\ff,\ffsq,\prec$
\end{algorithmic}
\end{algorithm}

\subsubsection{Backfitting Algorithm}

The backfitting algorithm iteratively refines a fit of $K$ factors and loadings, by updating them one at a time, at each update keeping the other loadings and factors fixed. The name comes from its connection with
the backfitting algorithm in \citet{Breiman1985},
specifically the fact that it involves iteratively
re-fitting to residuals.


\begin{algorithm}[H]
\caption{Backfitting algorithm for EBMF (rank $K$)} \label{alg:backfit}
\begin{algorithmic}[1]
\Require A data matrix $Y$ ($n \times p$)
\Require A function, ${\tt ebnm}(\bx,\bs) \rightarrow (\widebar{\btheta},\widebar{\btheta^2})$, that solves the EBNM problem \eqref{eqn:normalmeans}-\eqref{eqn:theta} and returns the first and second posterior moments \eqref{eqn:postmean}-\eqref{eqn:postmean2}.
\Require A function, ${\tt init}(Y) \rightarrow (\l_1,\dots,\l_K;\f_1,\dots,\f_K)$ that
provides initial estimates for the loadings and factors (e.g.~the greedy algorithm from Appendix \ref{app:greedy},
or a rank $K$ SVD).
\Require A function ${\tt single\_update}(Y,\ll,\ff,\llsq,\ffsq,k) \rightarrow(\ll_k,\llsq_k,\ff_k,\ffsq_k,\prec)$ implementing Algorithm \ref{alg:update}.
\State Initialize first moments $(\ll_1,\dots,\ll_K;\ff_1,\dots,\ff_K) \gets {\tt init}(Y)$.
\State Initialize second moments by squaring first moments: 
$\llsq_k \gets \ll_k^2; \ffsq_k \gets \ff_k^2$. 
[alternatively the {\tt init} function could provide these initial values].
\Repeat 
\For{$k=1,\dots,K$}
	\State $(\ll_k,\llsq_k,\ff_k,\ffsq_k,\prec) \gets
    {\tt single\_update}(Y,\ll,\llsq,\ff,\ffsq,k)$
\EndFor
\Until converged\\
\Return $\ll,\llsq,\ff,\ffsq,\prec$
\end{algorithmic}
\end{algorithm}

\subsection{Objective Function Computation}

The algorithms above all involve
updates that will increase (or, at least, not decrease) the objective function $F(\ql,\qf,\gl,\gf,\prec)$. However, 
these updates do not
require computing the objective function itself. In iterative algorithms
it can be helpful to compute
the objective function to monitor convergence (and as
a check on implementation). In this subsection we
describe how this can be done. In essence, this involves extending the solver of the 
EBNM problem to also
return the value of the $\log$-likelihood achieved in that problem (which is usually not difficult).

The objective function of the EBMF model is:
\begin{equation}
F(\ql,\qf,\gl,\gf,\prec) = E_{\ql,\qf}\log p(Y|\l,\f;\prec) + 
E_{\ql} \log \frac{\gl(\l)}{\ql(\l)} + E_{\qf} \log \frac{\gf(\f)}{\qf(\f)}
\end{equation}
The calculation of $E_{\ql,\qf}\log p(Y|\l,\f;\prec)$ is straightforward and $E_{\ql} \log \frac{\gl(\l)}{\ql(\l)}$ and $E_{\ql} \log \frac{\gl(\l)}{\ql(\l)}$ can be calculated using the log-likelihood of the EBNM model using the following Lemma \ref{lemma:obj}.


\begin{lemma} \label{lemma:obj}
Suppose $\hat{g},q$ solves the EBNM problem with data $(\bx,\bs)$: 
\begin{equation}
(\hat{g},q) = \ebnm (\bx,\bs),
\end{equation}
where $q:=(q_1,\dots,q_n)$ are the estimated posterior distributions of the normal means parameters $\theta_1,\dots,\theta_n$. Then 
\begin{equation}
E_q (\log(\prod_j \hat{g}(\theta_j)/\prod_j q_j(\theta_j)))
= l(\hat{g}; \bx,\bs) + \frac{1}{2} \sum_j \log (2\pi s_j^2) + (1/s_j^2) (x_j^2 + E_q (\theta_j^2) - 2x_j E_q(\theta_j) 
\end{equation}
where $l(\hat{g}; \bx,\bs)$ is the log of the likelihood for the normal means problem \eqref{eqn:lik-F-D}.
\end{lemma}

\begin{proof}
We have from \eqref{def_elbow}
\begin{align}
\FNM(\qt,\hat{g}) & = \int \qt(\btheta) \log \frac{p(\bx, \btheta | \hat{g})}{\qt(\btheta)} \, d\btheta \\
 & = \int \qt(\btheta) \log \frac{p(\bx | \btheta)\hat{g}(\btheta)}{\qt(\btheta)} \, d\btheta \\
 & = E_q(\log(\prod_j \hat{g}(\theta_j)/\prod_j q_j(\theta_j))) 
 - \frac{1}{2} E_q[\sum_j \log(2\pi s_j^2) + (1/s_j^2) (x_j-\theta_j)^2 ]
\end{align}
And the result follows from noting that $\FNM(\hat{q},\hat{g})= l(\hat{g})$.
\end{proof}



\subsection{Inference with Penalty Term} \label{sec:penalty}

Conceivably, in some settings one might like
to encourage solutions to the EBMF problem be sparser than the maximum-likelihood estimates for $\gl,\gf$ would produce. This could be done by extending the EBMF model to introduce a penalty term on the distributions $\gl,\gf$
so that the maximum likelihood estimates are replaced by maximizing a penalized likelihood.  We
are not advocating for this approach, but it is straightforward given existing machinery, and so we
document it here for completeness.

Let $h_{\l}(g_{\l})$ and $h_{\f}(g_{\f})$ denote penalty terms on $g_{\l}$ and $g_{\f}$, so the penalized log-likelihood would be:
\begin{eqnarray}
l(\gl, \gf, \prec)  &:=& \log [p(Y|\gl,\gf,\prec^2)] + h_{\l}(\gl) + h_{\f}(\gf)\\
& =& F(q,\gl,\gf,\prec^2) +  h_{\l}(\gl) + h_{\f}(\gf) + D_{KL}(q||p) \nonumber
\end{eqnarray}
where $F(q,\gl,\gf,\prec^2)$ and $ D_{KL}(q||p)$ are defined in~\eqref{elbow} and~\eqref{KLdistance}. And the corresponding penalized variational objective is: 
\begin{eqnarray}
\max F(q,\gl,\gf,\prec^2) +  h_{\l}(\gl) + h_{\f}(\gf).
\end{eqnarray}

It is straightforward to modify the algorithms above
to maximize this penalized objective: simply modify
the EBNM solvers to solve a corresponding penalized normal means problem. That is, instead of estimating the prior $g$
by maximum likelihood, the EBNM solver must now
maximize the penalized log-likelihood:
\begin{equation} \label{eqn:ghat2_reg}
\hat{g} = \arg \max_{g \in \G} l_\text{EBNM}(g) + h(g),
\end{equation} 
where $l_\text{EBNM}$ denote the log-likelihood for the EBNM problem. (The computation of the posterior distributions
given $\hat{g}$ is unchanged).

For example, the {\tt ashr} software 
\citep{stephens:2017} provides the option
to include a penalty on $g$ to encourage
overestimation of the size of the point mass on zero.  This penalty was introduced to ensure conservative behavior in False Discovery Rate applications of the normal means problem. It is unclear that such a penalty is desirable
in the matrix factorization application. However,
the above discussion shows
that using this penalty (e.g.~ within the {\tt ebnm} function used by the greedy or backfitting algorithms) can be thought of as solving a penalized version of the EBMF problem.

\section{Orthogonal Cross Validation} 
\label{OCValgorithm}

Cross-validation assessments involving ``holding out" (hiding) data from methods.
Here we introduce a novel approach to selecting the data to be held out,
which we call Orthogonal Cross Validation (OCV). Although not the main focus of our paper, we believe that OCV is a novel and appealing approach to selecting
hold-out data for factor models, e.g.~when using
CV to select an appropriate dimension $K$ for dimension reduction
methods, as in \cite{owen2016bi}.

Generic $k$-fold CV involves randomly dividing the data matrix into $k$ parts
and then, for each part, training methods on the other $k-1$ parts 
before assessing error on that part, as in Algorithm~\ref{alg:CV}.
\begin{algorithm}[H]
\caption{k-fold CV}\label{alg:CV}
\begin{algorithmic}[1]
\Procedure{k-fold cross validation}{}
\State randomly divide data matrix $Y$ into $Y_{(1)},\cdots,Y_{(k)}$ with ``hold-out" index $\Omega_{(1)},\cdots,\Omega_{(k)}$
\For {$i = 1, \cdots, k$}
\State take $Y_{(i)}$ as missing and run \flash{}
\State $ \hat{Y}_{(i)} = E [Y_{\Omega_{(i)}}| Y_{-\Omega_{(i)}}]$  
\State $s^2_i = ||\hat{Y}_{(i)}-Y_{\Omega_{(i)}}||_2^2$ 
\EndFor
\Return RMSE: $score = \sqrt{\frac{\sum_k s^2_k}{NP}}$ 
\EndProcedure
\end{algorithmic}
\end{algorithm}

The novel part of OCV is in how to choose the ``hold-out" pattern. 
We randomly divide the columns and rows into $k$ sets. and put these sets into $k$ orthogonal parts, and then take all $Y_{ij}$ with the chosen column and row indices as ``hold-out" $Y_{(i)}$. 

To illustrate this scheme, we take 3-fold CV as an example. We randomly divide the columns into 3 sets and the rows into 3 sets as well. The data matrix $Y$ is divided into 9 partition (by row and column permutation):

$$Y = \begin{pmatrix} Y_{11} & Y_{12} & Y_{13} \\ Y_{21} & Y_{22} & Y_{23} \\ Y_{31} & Y_{32} & Y_{33} \end{pmatrix}$$

Then $Y_{(1)} = \{  Y_{11},  Y_{22} , Y_{33}\}$, $Y_{(2)} = \{  Y_{12},  Y_{23} , Y_{31}\}$ and $Y_{(3)} = \{  Y_{13},  Y_{21} , Y_{32}\}$ are orthogonal to each other. Then the data matrix $Y$ is marked as:
$$Y = \begin{pmatrix} Y_{(1)} & Y_{(2)} & Y_{(3)} \\ Y_{(3)} & Y_{(1)} & Y_{(2)} \\ Y_{(2)} & Y_{(3)} & Y_{(1)} \end{pmatrix}$$

In OCV, each fold $k$, $Y_{(k)}$ contains equally balanced part of data matrix and includes all the row and column indices. This ensures that all $i$'s and $j$'s are included into each $Y_{-(k)}$.  In 3-fold OCV, we have:
\begin{align}
Y&= \begin{pmatrix} Y_{11} & Y_{12} & Y_{13}\\ Y_{21} & Y_{22}  & Y_{23} \\  Y_{31} & Y_{32}  & Y_{33}\end{pmatrix} = \left[ \begin{array}{c} L^{(1)} \\ L^{(2)} \\ L^{(3)}\end{array} \right] \times \left[ \begin{array}{ccc} F^{(1)} & F^{(2)} & F^{(3)} \end{array} \right] + E\\ 
&= \begin{pmatrix} Y_{(1)} & Y_{(2)} & Y_{(3)} \\ Y_{(3)} & Y_{(1)} & Y_{(2)} \\ Y_{(2)} & Y_{(3)} & Y_{(1)} \end{pmatrix} =  \begin{bmatrix} L^{(1)}F^{(1)} & L^{(1)}F^{(2)} & L^{(1)}F^{(3)} \\  L^{(2)}F^{(1)} & L^{(2)}F^{(2)} & L^{(2)}F^{(3)} \\  L^{(3)}F^{(1)} & L^{(3)}F^{(2)} & L^{(3)}F^{(3)}\end{bmatrix} + E 
\end{align}
where $Y_{(1)} = \{  Y_{11},  Y_{22} , Y_{33}\}$, $Y_{(2)} = \{  Y_{12},  Y_{23} , Y_{31}\}$ and $Y_{(3)} = \{  Y_{13},  Y_{21} , Y_{32}\}$. We can see for each ``hold-out" part, $Y_{(k)}$, $L^{(1)},L^{(2)},L^{(3)}$ and $F^{(1)},F^{(2)},F^{(3)}$ show up once and only once. 
In this sense the hold-out pattern is ``balanced".

\vskip 0.2in
\bibliography{FLASH}

\begin{thebibliography}{68}
\providecommand{\natexlab}[1]{#1}
\providecommand{\url}[1]{\texttt{#1}}
\expandafter\ifx\csname urlstyle\endcsname\relax
  \providecommand{\doi}[1]{doi: #1}\else
  \providecommand{\doi}{doi: \begingroup \urlstyle{rm}\Url}\fi

\bibitem[Argelaguet et~al.(2018)Argelaguet, Velten, Arnol, Dietrich, Zenz,
  Marioni, Buettner, Huber, and Stegle]{argelaguet2018multi}
Ricard Argelaguet, Britta Velten, Damien Arnol, Sascha Dietrich, Thorsten Zenz,
  John~C Marioni, Florian Buettner, Wolfgang Huber, and Oliver Stegle.
\newblock Multi-omics factor analysis—a framework for unsupervised
  integration of multi-omics data sets.
\newblock \emph{Molecular systems biology}, 14\penalty0 (6):\penalty0 e8124,
  2018.

\bibitem[Attias(1999)]{attias1999independent}
Hagai Attias.
\newblock Independent factor analysis.
\newblock \emph{Neural Computation}, 11\penalty0 (4):\penalty0 803--851, 1999.

\bibitem[Bai and Ng(2008)]{bai2008large}
Jushan Bai and Serena Ng.
\newblock \emph{Large Dimensional Factor Analysis}.
\newblock Now Publishers Inc, 2008.

\bibitem[Bhattacharya and Dunson(2011)]{bhattacharya2011sparse}
Anirban Bhattacharya and David~B Dunson.
\newblock Sparse {B}ayesian infinite factor models.
\newblock \emph{Biometrika}, pages 291--306, 2011.

\bibitem[Bishop(1999)]{bishop1999variational}
Christopher~M. Bishop.
\newblock Variational principal components.
\newblock In \emph{Ninth International Conference on Artificial Neural Networks
  (Conf. Publ. No. 470)}, volume~1, pages 509--514. IET, 1999.

\bibitem[Blei et~al.(2017)Blei, Kucukelbir, and McAuliffe]{Blei2016}
David~M Blei, Alp Kucukelbir, and Jon~D McAuliffe.
\newblock Variational inference: A review for statisticians.
\newblock \emph{Journal of the American Statistical Association}, 112\penalty0
  (518):\penalty0 859--877, 2017.

\bibitem[Bouchard et~al.(2015)Bouchard, Naradowsky, Riedel, Rockt{\"a}schel,
  and Vlachos]{bouchard-etal-2015-matrix}
Guillaume Bouchard, Jason Naradowsky, Sebastian Riedel, Tim Rockt{\"a}schel,
  and Andreas Vlachos.
\newblock Matrix and tensor factorization methods for natural language
  processing.
\newblock In \emph{Proceedings of the 53rd Annual Meeting of the Association
  for Computational Linguistics and the 7th International Joint Conference on
  Natural Language Processing: Tutorial Abstracts}, pages 16--18, Beijing,
  China, July 2015. Association for Computational Linguistics.
\newblock \doi{10.3115/v1/P15-5005}.
\newblock URL \url{https://www.aclweb.org/anthology/P15-5005}.

\bibitem[Breiman and Friedman(1985)]{Breiman1985}
Leo Breiman and Jerome~H. Friedman.
\newblock Estimating optimal transformations for multiple regression and
  correlation.
\newblock \emph{Journal of the American Statistical Association}, 80\penalty0
  (391):\penalty0 580--598, 1985.
\newblock \doi{10.1080/01621459.1985.10478157}.
\newblock URL
  \url{http://www.tandfonline.com/doi/abs/10.1080/01621459.1985.10478157}.

\bibitem[Carvalho et~al.(2008)Carvalho, Chang, Lucas, Nevins, Wang, and
  West]{Carvalho2008}
Carlos~M. Carvalho, Jeffrey Chang, Joseph~E. Lucas, Joseph~R. Nevins, Quanli
  Wang, and Mike West.
\newblock High-dimensional sparse factor modeling: Applications in gene
  expression genomics.
\newblock \emph{Journal of the American Statistical Association}, 103\penalty0
  (484):\penalty0 1438--1456, 2008.
\newblock ISSN 0162-1459.
\newblock \doi{10.1198/016214508000000869}.

\bibitem[Clyde and George(2000)]{Clyde2000}
Merlise Clyde and Edward~I George.
\newblock Flexible empirical {B}ayes estimation for wavelets.
\newblock \emph{Journal of the Royal Statistical Society Series B}, 62\penalty0
  (4):\penalty0 681--698, 2000.
\newblock \doi{10.1111/1467-9868.00257}.

\bibitem[Consortium et~al.(2015)]{gtex2015genotype}
GTEx Consortium et~al.
\newblock The genotype-tissue expression {(GTEx)} pilot analysis: Multitissue
  gene regulation in humans.
\newblock \emph{Science}, 348\penalty0 (6235):\penalty0 648--660, 2015.

\bibitem[Ding et~al.(2008)Ding, Li, and Jordan]{ding2008convex}
Chris~HQ Ding, Tao Li, and Michael~I Jordan.
\newblock Convex and semi-nonnegative matrix factorizations.
\newblock \emph{IEEE transactions on pattern analysis and machine
  intelligence}, 32\penalty0 (1):\penalty0 45--55, 2008.

\bibitem[Eckart and Young(1936)]{eckart.young.1936}
C~Eckart and G~Young.
\newblock The approximation of one matrix by another of lower rank.
\newblock \emph{Psychometrika}, 1:\penalty0 211--218, 1936.

\bibitem[Engelhardt and Stephens(2010)]{Engelhardt2010}
Barbara~E Engelhardt and Matthew Stephens.
\newblock Analysis of population structure: a unifying framework and novel
  methods based on sparse factor analysis.
\newblock \emph{PLoS Genetics}, 6\penalty0 (9):\penalty0 e1001117, sep 2010.

\bibitem[Fithian et~al.(2018)Fithian, Mazumder, et~al.]{fithian2018flexible}
William Fithian, Rahul Mazumder, et~al.
\newblock Flexible low-rank statistical modeling with missing data and side
  information.
\newblock \emph{Statistical Science}, 33\penalty0 (2):\penalty0 238--260, 2018.

\bibitem[Ford et~al.(1986)Ford, MacCallum, and Tait]{ford1986application}
J~Kevin Ford, Robert~C MacCallum, and Marianne Tait.
\newblock The application of exploratory factor analysis in applied psychology:
  A critical review and analysis.
\newblock \emph{Personnel Psychology}, 39\penalty0 (2):\penalty0 291--314,
  1986.

\bibitem[Fr{\"u}hwirth-Schnatter and Lopes(2018)]{fruhwirth2018sparse}
Sylvia Fr{\"u}hwirth-Schnatter and Hedibert~Freitas Lopes.
\newblock Sparse {B}ayesian factor analysis when the number of factors is
  unknown.
\newblock \emph{arXiv preprint arXiv:1804.04231}, 2018.

\bibitem[Gao et~al.(2013)Gao, Brown, and Engelhardt]{Gao2013}
Chuan Gao, Christopher~D Brown, and Barbara~E Engelhardt.
\newblock {A latent factor model with a mixture of sparse and dense factors to
  model gene gene expression data with confounding effects}.
\newblock \emph{arXiv:1310.4792v1}, 2013.

\bibitem[Gao et~al.(2016)Gao, McDowell, Zhao, Brown, and
  Engelhardt]{10.1371/journal.pcbi.1004791}
Chuan Gao, Ian~C. McDowell, Shiwen Zhao, Christopher~D. Brown, and Barbara~E.
  Engelhardt.
\newblock Context specific and differential gene co-expression networks via
  {B}ayesian biclustering.
\newblock \emph{PLoS Computational Biology}, 12\penalty0 (7):\penalty0 1--39,
  07 2016.
\newblock \doi{10.1371/journal.pcbi.1004791}.
\newblock URL \url{https://doi.org/10.1371/journal.pcbi.1004791}.

\bibitem[Ghahramani and Beal(2000)]{ghahramani2000variational}
Zoubin Ghahramani and Matthew~J Beal.
\newblock Variational inference for {B}ayesian mixtures of factor analysers.
\newblock In \emph{Advances in neural information processing systems}, pages
  449--455, 2000.

\bibitem[Girolami(2001)]{girolami2001variational}
Mark Girolami.
\newblock A variational method for learning sparse and overcomplete
  representations.
\newblock \emph{Neural Computation}, 13\penalty0 (11):\penalty0 2517--2532,
  2001.

\bibitem[Harper and Konstan(2016)]{harper2016movielens}
F~Maxwell Harper and Joseph~A Konstan.
\newblock The {M}ovielens datasets: History and context.
\newblock \emph{ACM Transactions on Interactive Intelligent Systems},
  5\penalty0 (4):\penalty0 19, 2016.

\bibitem[Hastie et~al.(2015)Hastie, Mazumder, Lee, and Zadeh]{hastie2015matrix}
Trevor Hastie, Rahul Mazumder, Jason~D Lee, and Reza Zadeh.
\newblock Matrix completion and low-rank {SVD} via fast alternating least
  squares.
\newblock \emph{Journal of Machine Learning Research}, 16\penalty0
  (1):\penalty0 3367--3402, 2015.

\bibitem[Hochreiter et~al.(2010)Hochreiter, Bodenhofer, Heusel, Mayr,
  Mitterecker, Kasim, Khamiakova, Van~Sanden, Lin, Talloen, Bijnens, Göhlmann,
  Shkedy, and Clevert]{hochreiter2010FABIA}
Sepp Hochreiter, Ulrich Bodenhofer, Martin Heusel, Andreas Mayr, Andreas
  Mitterecker, Adetayo Kasim, Tatsiana Khamiakova, Suzy Van~Sanden, Dan Lin,
  Willem Talloen, Luc Bijnens, Hinrich W.~H. Göhlmann, Ziv Shkedy, and
  Djork-Arné Clevert.
\newblock {FABIA: factor analysis for bicluster acquisition}.
\newblock \emph{Bioinformatics}, 26\penalty0 (12):\penalty0 1520--1527, 04
  2010.
\newblock ISSN 1367-4803.
\newblock \doi{10.1093/bioinformatics/btq227}.
\newblock URL \url{https://doi.org/10.1093/bioinformatics/btq227}.

\bibitem[Hore et~al.(2016)Hore, Vi{\~n}uela, Buil, Knight, McCarthy, Small, and
  Marchini]{hore2016tensor}
Victoria Hore, Ana Vi{\~n}uela, Alfonso Buil, Julian Knight, Mark~I McCarthy,
  Kerrin Small, and Jonathan Marchini.
\newblock Tensor decomposition for multiple-tissue gene expression experiments.
\newblock \emph{Nature Genetics}, 48\penalty0 (9):\penalty0 1094--1100, 2016.

\bibitem[Jaakkola and Jordan(2000)]{Jordan2000}
Tommi~S. Jaakkola and Michael~I. Jordan.
\newblock {{B}ayesian parameter estimation via variational methods.}
\newblock \emph{Statistics and Computing}, 10:\penalty0 25--37, 2000.

\bibitem[Johnstone and Silverman(2005{\natexlab{a}})]{Johnstone2005Empirical}
Iain~M. Johnstone and Bernard~W. Silverman.
\newblock Empirical {B}ayes selection of wavelet thresholds.
\newblock \emph{The Annals of Statistics}, 33\penalty0 (4):\penalty0
  1700--1752, 2005{\natexlab{a}}.

\bibitem[Johnstone and
  Silverman(2005{\natexlab{b}})]{johnstone2005ebayesthresh}
Iain~M Johnstone and Bernard~W Silverman.
\newblock E{B}ayesthresh: R and s-plus programs for empirical {B}ayes
  thresholding.
\newblock \emph{J. Statist. Soft}, 12:\penalty0 1--38, 2005{\natexlab{b}}.

\bibitem[Johnstone et~al.(2004)Johnstone, Silverman,
  et~al.]{johnstone2004needles}
Iain~M Johnstone, Bernard~W Silverman, et~al.
\newblock Needles and straw in haystacks: Empirical {B}ayes estimates of
  possibly sparse sequences.
\newblock \emph{The Annals of Statistics}, 32\penalty0 (4):\penalty0
  1594--1649, 2004.

\bibitem[Jolliffe et~al.(2003)Jolliffe, Trendafilov, and
  Uddin]{jolliffe2003modified}
Ian~T Jolliffe, Nickolay~T Trendafilov, and Mudassir Uddin.
\newblock A modified principal component technique based on the lasso.
\newblock \emph{Journal of Computational and Graphical Statistics}, 12\penalty0
  (3):\penalty0 531--547, 2003.

\bibitem[Josse et~al.(2018)Josse, Sardy, and Wager]{josse2018denoiser}
Julie Josse, Sylvain Sardy, and Stefan Wager.
\newblock denoise{R}: A package for low rank matrix estimation, 2018.

\bibitem[Kaufmann and Schumacher(2017)]{kaufmann2017identifying}
Sylvia Kaufmann and Christian Schumacher.
\newblock Identifying relevant and irrelevant variables in sparse factor
  models.
\newblock \emph{Journal of Applied Econometrics}, 32\penalty0 (6):\penalty0
  1123--1144, 2017.

\bibitem[Klami(2015)]{klami2015polya}
Arto Klami.
\newblock Polya-gamma augmentations for factor models.
\newblock In \emph{Asian Conference on Machine Learning}, pages 112--128, 2015.

\bibitem[Knowles and Ghahramani(2011)]{Knowles2011}
David Knowles and Zoubin Ghahramani.
\newblock Nonparametric {B}ayesian sparse factor.
\newblock \emph{Annals of Applied Statistics}, 5\penalty0 (2B):\penalty0
  1534--1552, 2011.
\newblock \doi{10.1214/10-AOAS435}.

\bibitem[Koenker and Mizera(2014{\natexlab{a}})]{koenker2014JASA}
Roger Koenker and Ivan Mizera.
\newblock Convex optimization, shape constraints, compound decisions, and
  empirical {B}ayes rules.
\newblock \emph{Journal of the American Statistical Association}, 109\penalty0
  (506):\penalty0 674--685, 2014{\natexlab{a}}.

\bibitem[Koenker and Mizera(2014{\natexlab{b}})]{koenker2014convex}
Roger Koenker and Ivan Mizera.
\newblock Convex optimization in {R}.
\newblock \emph{Journal of Statistical Software}, 60\penalty0 (5):\penalty0
  1--23, 2014{\natexlab{b}}.
\newblock \doi{10.18637/jss.v060.i05}.

\bibitem[Kolda and Bader(2009)]{kolda2009tensor}
Tamara~G Kolda and Brett~W Bader.
\newblock Tensor decompositions and applications.
\newblock \emph{SIAM review}, 51\penalty0 (3):\penalty0 455--500, 2009.

\bibitem[Lee and Seung(1999)]{lee:1999}
D~D Lee and H~S Seung.
\newblock Learning the parts of objects by non-negative matrix factorization.
\newblock \emph{Nature}, 401\penalty0 (6755):\penalty0 788--791, 1999.
\newblock \doi{10.1038/44565}.

\bibitem[Lefran{\c{c}}ais et~al.(2017)Lefran{\c{c}}ais, Ortiz-mu{\~{n}}oz,
  Caudrillier, Mallavia, Liu, Sayah, Thornton, Headley, David, Coughlin,
  Krummel, Leavitt, Passegu{\'{e}}, and Looney]{Lefrancais2017}
Emma Lefran{\c{c}}ais, Guadalupe Ortiz-mu{\~{n}}oz, Axelle Caudrillier,
  Be{\~{n}}at Mallavia, Fengchun Liu, David~M Sayah, Emily~E Thornton, Mark~B
  Headley, Tovo David, Shaun~R Coughlin, Matthew~F Krummel, Andrew~D Leavitt,
  Emmanuelle Passegu{\'{e}}, and Mark~R Looney.
\newblock {The lung is a site of platelet biogenesis and a reservoir for
  haematopoietic progenitors}.
\newblock \emph{Nature}, 544\penalty0 (7648):\penalty0 105--109, 2017.
\newblock \doi{10.1038/nature21706}.

\bibitem[Lim and Teh(2007)]{lim2007variational}
Yew~Jin Lim and Yee~Whye Teh.
\newblock Variational {B}ayesian approach to movie rating prediction.
\newblock In \emph{Proceedings of KDD cup and workshop}, volume~7, pages
  15--21. Citeseer, 2007.

\bibitem[Mayrink et~al.(2013)Mayrink, Lucas, et~al.]{mayrink2013sparse}
Vinicius~Diniz Mayrink, Joseph~Edward Lucas, et~al.
\newblock Sparse latent factor models with interactions: Analysis of gene
  expression data.
\newblock \emph{The Annals of Applied Statistics}, 7\penalty0 (2):\penalty0
  799--822, 2013.

\bibitem[Mazumder et~al.(2010)Mazumder, Hastie, and
  Tibshirani]{mazumder2010spectral}
Rahul Mazumder, Trevor Hastie, and Robert Tibshirani.
\newblock Spectral regularization algorithms for learning large incomplete
  matrices.
\newblock \emph{Journal of Machine Learning Research}, 11\penalty0
  (Aug):\penalty0 2287--2322, 2010.

\bibitem[Nakajima and Sugiyama(2011)]{nakajima2011theoretical}
Shinichi Nakajima and Masashi Sugiyama.
\newblock Theoretical analysis of {B}ayesian matrix factorization.
\newblock \emph{Journal of Machine Learning Research}, 12:\penalty0 2583--2648,
  2011.

\bibitem[Nakajima et~al.(2013)Nakajima, Sugiyama, Babacan, and
  Tomioka]{nakajima2013global}
Shinichi Nakajima, Masashi Sugiyama, S~Derin Babacan, and Ryota Tomioka.
\newblock Global analytic solution of fully-observed variational bayesian
  matrix factorization.
\newblock \emph{Journal of Machine Learning Research}, 14\penalty0
  (Jan):\penalty0 1--37, 2013.

\bibitem[Owen and Wang(2016)]{owen2016bi}
Art~B Owen and Jingshu Wang.
\newblock Bi-cross-validation for factor analysis.
\newblock \emph{Statistical Science}, 31\penalty0 (1):\penalty0 119--139, 2016.

\bibitem[Pournara and Wernisch(2007)]{Pournara2007}
Iosifina Pournara and Lorenz Wernisch.
\newblock {Factor analysis for gene regulatory networks and transcription
  factor activity profiles.}
\newblock \emph{BMC bioinformatics}, 8:\penalty0 61, 2007.
\newblock ISSN 1471-2105.
\newblock \doi{10.1186/1471-2105-8-61}.

\bibitem[Raiko et~al.(2007)Raiko, Ilin, and Karhunen]{raiko2007principal}
Tapani Raiko, Alexander Ilin, and Juha Karhunen.
\newblock Principal component analysis for large scale problems with lots of
  missing values.
\newblock In \emph{European Conference on Machine Learning}, pages 691--698.
  Springer, 2007.

\bibitem[Ro{\v{c}}kov{\'a} and George(2016)]{rovckova2016fast}
Veronika Ro{\v{c}}kov{\'a} and Edward~I George.
\newblock Fast {B}ayesian factor analysis via automatic rotations to sparsity.
\newblock \emph{Journal of the American Statistical Association}, 111\penalty0
  (516):\penalty0 1608--1622, 2016.

\bibitem[Rubin(1976)]{rubin1976inference}
Donald~B Rubin.
\newblock Inference and missing data.
\newblock \emph{Biometrika}, 63\penalty0 (3):\penalty0 581--592, 1976.

\bibitem[Rubin and Thayer(1982)]{rubin1982algorithms}
Donald~B Rubin and Dorothy~T Thayer.
\newblock {EM} algorithms for {ML} factor analysis.
\newblock \emph{Psychometrika}, 47\penalty0 (1):\penalty0 69--76, 1982.

\bibitem[Sabatti and James(2005)]{sabatti2005bayesian}
Chiara Sabatti and Gareth~M James.
\newblock {B}ayesian sparse hidden components analysis for transcription
  regulation networks.
\newblock \emph{Bioinformatics}, 22\penalty0 (6):\penalty0 739--746, 2005.

\bibitem[Seeger and Bouchard(2012)]{seeger2012fast}
Matthias Seeger and Guillaume Bouchard.
\newblock Fast variational {B}ayesian inference for non-conjugate matrix
  factorization models.
\newblock In \emph{Artificial Intelligence and Statistics}, pages 1012--1018,
  2012.

\bibitem[Srivastava et~al.(2017)Srivastava, Engelhardt, and
  Dunson]{srivastava2017expandable}
Sanvesh Srivastava, Barbara~E Engelhardt, and David~B Dunson.
\newblock Expandable factor analysis.
\newblock \emph{Biometrika}, 104\penalty0 (3):\penalty0 649--663, 2017.

\bibitem[Stegle et~al.(2010)Stegle, Parts, Durbin, and
  Winn]{stegle2010bayesian}
Oliver Stegle, Leopold Parts, Richard Durbin, and John Winn.
\newblock A {B}ayesian framework to account for complex non-genetic factors in
  gene expression levels greatly increases power in eqtl studies.
\newblock \emph{PLoS Computational Biology}, 6\penalty0 (5):\penalty0 e1000770,
  2010.

\bibitem[Stegle et~al.(2012)Stegle, Parts, Piipari, Winn, and
  Durbin]{stegle2012using}
Oliver Stegle, Leopold Parts, Matias Piipari, John Winn, and Richard Durbin.
\newblock Using probabilistic estimation of expression residuals ({PEER}) to
  obtain increased power and interpretability of gene expression analyses.
\newblock \emph{Nature Protocols}, 7\penalty0 (3):\penalty0 500--507, 2012.

\bibitem[Stein-O’Brien et~al.(2018)Stein-O’Brien, Arora, Culhane, Favorov,
  Garmire, Greene, Goff, Li, Ngom, Ochs, et~al.]{stein2018enter}
Genevieve~L Stein-O’Brien, Raman Arora, Aedin~C Culhane, Alexander~V Favorov,
  Lana~X Garmire, Casey~S Greene, Loyal~A Goff, Yifeng Li, Aloune Ngom,
  Michael~F Ochs, et~al.
\newblock Enter the matrix: factorization uncovers knowledge from omics.
\newblock \emph{Trends in Genetics}, 34\penalty0 (10):\penalty0 790--805, 2018.

\bibitem[Stephens(2017)]{stephens:2017}
Matthew Stephens.
\newblock False discovery rates: a new deal.
\newblock \emph{Biostatistics}, 18\penalty0 (2):\penalty0 275--294, Apr 2017.
\newblock \doi{10.1093/biostatistics/kxw041}.

\bibitem[Thomas et~al.(1985)Thomas, Siemiatycki, Dewar, Robins, Goldberg, and
  Armstrong]{thomas:1985}
D~C Thomas, J~Siemiatycki, R~Dewar, J~Robins, M~Goldberg, and B~G Armstrong.
\newblock The problem of multiple inference in studies designed to generate
  hypotheses.
\newblock \emph{American Journal of Epidemiology}, 122\penalty0 (6):\penalty0
  1080--95, Dec 1985.

\bibitem[Tipping(2001)]{tipping2001sparse}
Michael~E Tipping.
\newblock Sparse {B}ayesian learning and the relevance vector machine.
\newblock \emph{Journal of Machine Learning Research}, 1\penalty0
  (Jun):\penalty0 211--244, 2001.

\bibitem[Titsias and L\'{a}zaro-Gredilla(2011)]{titsias2011spike}
Michalis~K. Titsias and Miguel L\'{a}zaro-Gredilla.
\newblock Spike and slab variational inference for multi-task and multiple
  kernel learning.
\newblock In J.~Shawe-Taylor, R.~S. Zemel, P.~L. Bartlett, F.~Pereira, and
  K.~Q. Weinberger, editors, \emph{Advances in Neural Information Processing
  Systems 24}, pages 2339--2347. Curran Associates, Inc., 2011.
\newblock URL
  \url{http://papers.nips.cc/paper/4305-spike-and-slab-variational-inference-for-multi-task-and-multiple-kernel-learning.pdf}.

\bibitem[Wang(2017)]{wang2017}
Wei Wang.
\newblock \emph{Applications of Adaptive Shrinkage in Multiple Statistical
  Problems}.
\newblock PhD thesis, The University of Chicago, 2017.

\bibitem[West(2003)]{West2003}
Mike West.
\newblock {{B}ayesian factor regression models in the "large p, small n"
  paradigm}.
\newblock \emph{{B}ayesian Statistics 7 - Proceedings of the Seventh Valencia
  International Meeting}, pages 723--732, 2003.
\newblock ISSN 08966273.
\newblock \doi{10.1.1.18.3036}.

\bibitem[Wipf and Nagarajan(2008)]{wipf2008new}
David~P. Wipf and Srikantan~S. Nagarajan.
\newblock A new view of automatic relevance determination.
\newblock In J.~C. Platt, D.~Koller, Y.~Singer, and S.~T. Roweis, editors,
  \emph{Advances in Neural Information Processing Systems 20}, pages
  1625--1632. Curran Associates, Inc., 2008.
\newblock URL
  \url{http://papers.nips.cc/paper/3372-a-new-view-of-automatic-relevance-determination.pdf}.

\bibitem[Witten et~al.(2009)Witten, Tibshirani, and Hastie]{Witten2009}
Daniela~M. Witten, Robert Tibshirani, and Trevor Hastie.
\newblock {A penalized matrix decomposition, with applications to sparse
  principal components and canonical correlation analysis}.
\newblock \emph{Biostatistics}, 10\penalty0 (3):\penalty0 515--534, 2009.
\newblock \doi{10.1093/biostatistics/kxp008}.

\bibitem[Xing et~al.(2016)Xing, Carbonetto, and Stephens]{xing2016smoothing}
Zhengrong Xing, Peter Carbonetto, and Matthew Stephens.
\newblock Smoothing via adaptive shrinkage (smash): denoising {P}oisson and
  heteroskedastic {G}aussian signals.
\newblock \emph{arXiv preprint arXiv:1605.07787}, 2016.

\bibitem[Yang et~al.(2014)Yang, Ma, and Buja]{yang2014sparse}
Dan Yang, Zongming Ma, and Andreas Buja.
\newblock A sparse singular value decomposition method for high-dimensional
  data.
\newblock \emph{Journal of Computational and Graphical Statistics}, 23\penalty0
  (4):\penalty0 923--942, 2014.

\bibitem[Zhao et~al.(2018)Zhao, Engelhardt, Mukherjee, and
  Dunson]{zhao2018fast}
Shiwen Zhao, Barbara~E Engelhardt, Sayan Mukherjee, and David~B Dunson.
\newblock Fast moment estimation for generalized latent {D}irichlet models.
\newblock \emph{Journal of the American Statistical Association}, pages 1--13,
  2018.

\bibitem[Zou et~al.(2006)Zou, Hastie, and Tibshirani]{H.Zou2006}
Hui Zou, Trevor Hastie, and Robert Tibshirani.
\newblock Sparse principal component analysis.
\newblock \emph{Journal of Computational and Graphical Statistics}, 15\penalty0
  (2):\penalty0 265--286, 2006.
\newblock ISSN 1061-8600.
\newblock \doi{10.1198/106186006X113430}.

\end{thebibliography}

\end{document}